%% file: main.tex
\documentclass[sigconf]{acmart}
\input{macros}

\begin{document}
\title{Towards Measuring Membership Privacy} 

\thanks{This work was supported in part by NSF CNS grants 13-30491 and 14-08944. The views expressed are those of the authors only.}

\author{Yunhui Long}
\affiliation{%
  \institution{University of Illinois at Urbana-Champaign}
}
\email{ylong4@illinois.edu}

\author{Vincent Bindschaedler}
\affiliation{%
	\institution{University of Illinois at Urbana-Champaign}
}
\email{bindsch2@illinois.edu}

\author{Carl A. Gunter}
\affiliation{%
	\institution{University of Illinois at Urbana-Champaign}
}
\email{cgunter@illinois.edu}

\begin{abstract}
\input{abstract}
\end{abstract}






\thispagestyle{plain}
\pagestyle{plain}

\maketitle

\input{introduction} 
\input{threat_model}
\input{background}
\input{dtp}
\input{case_studies}
\input{dtp_measure}

\input{related_work}
\input{drop_dtp}
\input{discussion}

\input{open_questions}
\input{conclusion}

\bibliographystyle{acm}
\bibliography{reference}

\appendix
\input{appendix}

\end{document}

%% file: macros.tex
\usepackage{amsmath}
\usepackage{algorithmic}

\usepackage{colortbl}
\usepackage{graphicx,subcaption}
\usepackage{hhline}

\usepackage{stfloats}

\usepackage{url}

\usepackage{amsfonts}
\usepackage{amsthm}

\usepackage{cleveref}

\usepackage{algorithmic}
\usepackage[]{algorithm2e}

\usepackage{mathtools}

\usepackage[justification=centering]{caption}

\newcommand*\rfrac[2]{{}^{#1}\!/_{#2}}
\newcommand{\p}{\textrm{p}}

\newcommand{\DTP}{\textrm{DTP}}

\newcommand{\A}{\mathcal{A}}

\newcommand{\PDTP}{\textrm{PDTP}}
\newcommand*\mean[1]{\overline{#1}}

\DeclarePairedDelimiterX{\infdivx}[2]{(}{)}{%
  #1\;\delimsize\|\;#2%
}
\newcommand{\kldiv}{D_{\textrm{KL}}\infdivx}

\newcommand{\paragraphe}[1]{\vspace{0.75ex}\noindent{\em #1} }

\newcommand{\paragraphbe}[1]{\vspace{0.75ex}\noindent{\bf \em #1} }

%% file: abstract.tex
Machine learning models are increasingly made available to the masses through public query interfaces. Recent academic work has demonstrated that malicious users who can query such models are able to infer sensitive information about records within the training data. Differential privacy can thwart such attacks, but not all models can be readily trained to achieve this guarantee or to achieve it with acceptable utility loss. As a result, if a model is trained without differential privacy guarantee, little is known or can be said about the privacy risk of releasing it.

In this work, we investigate and analyze membership attacks to understand why and how they succeed. Based on this understanding, we propose \emph{Differential Training Privacy} (DTP), an empirical metric to estimate the privacy risk of publishing a classier when methods such as differential privacy cannot be applied. DTP is a measure on a classier with respect to its training dataset, and we show that calculating DTP is efficient in many practical cases. We empirically validate DTP using state-of-the-art machine learning models such as neural networks trained on real-world datasets. Our results show that DTP is highly predictive of the success of membership attacks and therefore reducing DTP also reduces the privacy risk. We advocate for DTP to be used as part of the decision-making process when considering publishing a classifier. To this end, we also suggest adopting the DTP-$1$ hypothesis: if a classifier has a DTP value above $1$, it should not be published. 

%% file: introduction.tex
\vspace{-6pt}
\section{Introduction}
\label{sec:intro}
Machine learning models are widely used to extract useful information from large datasets and support many popular Internet services. Companies like Amazon~\cite{amazon2016ml}, Google~\cite{google2016ml}, and Microsoft~\cite{microsoft2016ml} have started to provide Machine Learning-as-a-Service (MLaaS). Data owners upload their data and obtain black-box access to a classification model which can be queried through an API.

Recent academic work has pointed out several security and privacy issues with this MLaaS paradigm. Models can be stolen or reverse engineered~\cite{tramer2016stealing}, sensitive population-level information can be inferred~\cite{fredrikson2014privacy,fredrikson2015model}; even the corresponding training datasets can be targeted by inference attacks. In particular, Shokri et al.~\cite{shokri2016membership} propose a membership attack to infer sensitive information about individuals whose data records are part of the training dataset.

Membership attacks are not new or specific to MLaaS and are known credible threats in various contexts such as in Genome-Wide Association Studies (GWAS)~\cite{homer2008resolving}. In principle, membership attacks are easily thwarted by ensuring that the model is trained using a differentially private process. Unfortunately, it is not always feasible to use a learning algorithm that satisfies differential privacy. Some classification models cannot readily be trained in this way, or doing so may come at the cost of an unacceptable utility loss.

This issue is exacerbated by the fact that if a classification model is trained without differential privacy, then little is known about its membership privacy risk; there is a gap in our ability to analyze the risk. At the same time, there is no reason to believe that all classification algorithms leak the exact same amount about their training datasets --- a model that is badly overfitted has the potential to leak more than one which is not. Yet there is currently no framework or principled way to measure this in practice.

In this paper, we investigate why and how membership inference attacks succeed. We derive a general attack framework and perform experiments on state-of-the-art classifiers trained on real-world datasets. Our goal is to design a metric which reflects membership privacy risk and can easily be calculated on a classifier.

We identify a simple measure called {\em Differential Training Privacy (DTP)} which quantifies the risk of membership inference of a record with respect to a classifier and its training data; the higher the DTP value, the higher the risk. We extend DTP to a metric over the classifier, by computing the DTP value of all records in the training dataset and taking the maximum---the {\em worst case} risk. DTP is not a substitute for a differentially private learning algorithm. Rather DTP provides an objective basis for decision making. For example, when two classifiers exhibit similar performance, it is preferable to publish the one with the lowest DTP.

Informally, given a classifier $\mathcal{A}(T)$ trained on a dataset $T$, the membership leakage of a record $t \in T$ is quantified by comparing that classifier's predictions to those of a classifier trained without record $t$, i.e., $\mathcal{A}(T\setminus\{t\})$. We assume black-box access, so an adversary can only learn information by querying the classifier. In this setting, membership attacks are predicated on distinguishing whether the classifier was trained on $T$ or $T\setminus \{t\}$. This is what DTP measures. Differences in predictions can occur for any query, but we initially focus on {\em direct attacks} which expect the maximum difference to be observed when querying features of $t$.

We provide experimental validation of direct attacks on both traditional classifiers such as naive Bayes and logistic regression and state-of-the-art models such as neural networks trained on two real-world datasets: a purchase dataset containing the shopping history of 300,000 individuals, and the popular UCI Adult dataset. Specifically, we perform several membership inference attacks on these classifiers, including the most effective attack known. Results suggest that DTP is a powerful predictor of the accuracy of direct attacks. Concretely, for neural networks learned on the purchase dataset, the Pearson correlation coefficient of the maximum membership attack accuracy with DTP is $0.8936$. For classifiers with DTP-values under $0.5$, none of the attacks we performed ever inferred membership status of any individual with accuracy greater than $66.5\%$ (baseline: $50\%$). By comparison our attacks almost always have over $90\%$ accuracy when DTP is larger than $4$. 

Although we do not know of any practical {\em indirect attacks}---which query the classifier for features other than those of $t$---we cannot exclude the possibility that such attacks may outperform direct ones. In fact, we produce a counter-example that this is indeed possible. We explore whether this situation occurs for known classification algorithms and derive a set of theoretical results, including a simple criterion (for classifiers) called {\em training stability}. For algorithms which satisfy training stability, the direct attack is always superior to {\em any} indirect attack. 

\paragraphbe{Contributions.}
We propose Differential Training Privacy (DTP) to quantify membership inference risk of publishing a classifier. DTP is used to inform the decision of whether to release a classifier when techniques to achieve differential privacy cannot be employed. We advocate for the DTP-$1$ hypothesis: if a classifier has a DTP value above $1$, it should not be published. We test this hypothesis by designing effective attacks on records with DTP greater than $1$ based on different classifiers and datasets. 

We present a general membership attack framework and evaluate three types of attacks on several classifiers trained on two real-world datasets, including the most effective attack known prior to this work---which we improve upon. Our empirical study of the relationship between the accuracy of membership attacks and DTP, reveals the latter to be a powerful predictor of the former.

We establish training stability as an important desideratum for classifiers, and prove that naive Bayes, random decision trees, and linear statistical queries satisfy it but $k$-NN does not.

%% file: threat_model.tex

\section{Problem Statement}
Consider a data owner with a dataset $D$. We assume that the dataset is divided into multiple classes and each record in the dataset consists of a class label and a vector of features. We also assume that the data owner randomly partitions the dataset into a training set and a validation set. A machine learning model is trained on the training set. The model captures the correlations between the features and the class labels. It takes a feature vector as input and outputs a vector of probabilities for each class. 

Suppose the data owner wants to make the model available for public queries. That is, he intends to allow anyone to submit a feature vector and get predictions from the model. In this paper, we want to estimate the membership inference risk of releasing the model based on simple measurements and theoretical analysis. Specifically, given a machine learning model, we want to answer the following questions: \emph{Is there a privacy risk if the model is open to public queries? Are certain models riskier than others? Which records have higher privacy risks?}

We consider the privacy risk in the setting of membership privacy. The privacy risk of a record is estimated by the adversary's ability to infer whether the record is a part of the training set. We estimate this risk under the assumption that the machine learning models are trained by trusted parties. That is, we expect the parties training and releasing the model have the goal of protecting the training set to the extent this is practical and will therefore not be motivated to create a covert channel by embedding private information into the model's predictions. Under this assumption, we inspect the risk of accidental privacy leakage---the risk that a machine learning algoirthm would accidentally learn too much information about an individual record during the training process.  

To estimate privacy risk under a strong adversary, we assume that the adversary knows all the records in $D$, the size of the training set sampled from $D$, and the machine learning algorithm used to train the model. We assume the training set to be uniformly sampled from $D$. Therefore, each record in $D$ is equally likely to be included in the training set, and the adversary does not know which records are included. The goal of the adversary is to perform a membership inference attack by querying the machine learning model. That is, given a particular target record $t \in D$, he wants to infer whether $t$ is used to train the model he queries.  

\paragraphbe{Notations.}
Let $T$ be the training set of the model. Since $T$ is randomly sampled from $D$, we call $D$ the \emph{candidate set} for T.

We define $X^m$ to be a set of all possible features and $Y$ to be the set of all possible class labels: $\{y_1, y_2, \dots, y_k\}$. Each record $\mathbf{z} \in D$ can be divided into two parts: the feature vector $\mathbf{x} \in X^m$ and the class label $y \in Y$. Let $\A$ be the classification algorithm and $c = \A(T)$ be the output classifier. We assume that for each query $\mathbf{x}$, $c(\mathbf{x})$ returns a vector of conditional probability of all class labels $y \in \{y_1, y_2, \dots, y_k\}$ given feature vector $\mathbf{x}$. We use $p_{c}(y \mid \mathbf{x})$ to represent the conditional probability of class label $y$ given feature $\mathbf{x}$, predicted by classifier $c$. That is, $c(\mathbf{x}) = \left(\p_{c}(y_1 \mid \mathbf{x}), \p_{c}(y_2 \mid \mathbf{x}), \dots, \p_{c}(y_k \mid \mathbf{x})\right)$.
For classifiers that do not directly provide predicted probabilities, these can be obtained through normalization over the class labels.

In membership inferences, the adversary wants to infer whether a specific record $t = (\mathbf{x}^{(t)}, y^{(t)}) \in D$ is part of the training set $T$. We call $t$ the \emph{target record} of the attack. 

There are two approaches that an adversary can take to perform a membership inference attack on a target record $t$. He can launch a \emph{direct attack} by querying the features of the target record $\mathbf{x}^{(t)}$. Or, he can perform an \emph{indirect attack} by querying some feature vector $\mathbf{x} \neq \mathbf{x}^{(t)}$. Intuitively, a direct attack should have better performance than an indirect attack because querying the features of $t$ should give more information about $t$ compared to querying other features. In this paper, we first study the membership privacy risk under this assumption. In section~\ref{sec:measure_dtp}, we test this assumption by analyzing some commonly used classifiers. 

%% file: background.tex
\section{Background on Membership Attacks}
\label{sec:background}

In this section, we briefly review the membership attacks proposed in prior work~\cite{shokri2016membership}. In a membership attack, an adversary is given black-box access to a target classifier $\A(T)$ and wants to infer whether a particular record $t$ is included in the training set $T$. 

The membership attack proposed by Shokri et al.~\cite{shokri2016membership} is a direct attack. The adversary queries the target record $t$ and uses the target classifiers' predictions on $t$ to infer the membership status of $t$. The authors transform the membership attack into a classification task. For each record $t$, there are two possible classes: class label ``in'' represents that the record is in the training set, and class label ``out'' represents that the record is \emph{not} in the training set. The features in the membership classification task are the original attributes of $t$ and the target classifier's predictions on $t$. The adversary trains an attack classifier for the membership classification task and uses it to infer the membership of a target record.

To create training sets for the membership classification task, the authors introduce ``shadow'' training techniques. Concretely, the adversary trains multiple ``shadow models'' using the same machine learning algorithm $\A$ on records sampled from the same population as $T$. These shadow models are used to simulate the behavior of the target model and generate a set of training records with labeled membership information. Specifically, the adversary queries each shadow model with two sets of records: the training set of the shadow model and a disjoint test set. For each record, a new feature vector is generated by concatenating the record's original attributes with the shadow classifier's predictions on that record. A new class label is created to reflect membership, i.e., ``in'' for records in the training set, ``out'' for records in the test set. 

Using the labeled dataset, the adversary trains a neural network as ``attack'' classifier and uses it to infer the membership of a target record $t$, given black-box access to a classifier $\A(T)$.

In this paper, we reproduce this membership attack and design two new membership attacks based on shadow training techniques. One of the new membership attacks has much better performance compared to the attack proposed in \cite{shokri2016membership}. We use these attacks to validate DTP's ability to measure the membership privacy risk of a target record $t$ with respect to a classifier $\A(T)$.

%% file: dtp.tex
\section{Measuring Membership Privacy}

\subsection{Understanding Membership Attacks}

In a membership attack on a classifier $c = \A(T)$, an adversary tries to infer whether a target record $t = (\mathbf{x}^{(t)}, y^{(t)})$ is in the training set of $c$. The adversary can submit any query $\mathbf{x}$ and analyze the returned results $\mathbf{q} = c(\mathbf{x})$. With some background knowledge about the behavior of the classification algorithm and the candidate dataset $D$, an adversary can estimate the following two conditional probabilities: $\Pr \left[ t \in T \mid c(\mathbf{x}) = \mathbf{q} \right]$ and $\Pr \left[t \notin T \mid c(\mathbf{x}) = \mathbf{q} \right]$, which he uses to infer whether $t$ is a member of the training set.

The accuracy of a membership attack depends on two properties. First, it depends on the adversary's ability in correctly estimating $\Pr \left[ t \in T \mid c(\mathbf{x}) = \mathbf{q} \right]$ and $\Pr \left[t \notin T \mid c(\mathbf{x}) = \mathbf{q} \right]$. An adversary with a stronger background knowledge has more accurate estimates and thus obtains higher accuracy. Second, it depends on the properties of the target record $t$, the training dataset $T$, and the classification algorithm $\A$. Intuitively, a membership attack will have higher accuracy on overfitted classifiers because these capture statistical peculiarities of the training dataset that do not generalize to the whole population. Also, outliers may be more vulnerable to membership attacks. 

Based on this intuition, we want to estimate the membership attacks risk by measuring the generalizability of a classifier. Specifically, we calculate the maximum change in a classifier's predictions when the target record is removed from the training set.

\subsection{Differential Training Privacy}
We propose a measure of membership privacy risk of a target record $t$ with respect to classification algorithm $\A$ and training set $T$.

\begin{definition} [Differential Training Privacy]
\label{def:dtp}  
A record $t \in T$ is \emph{$\epsilon$-differentially training private} ($\epsilon$-DTP) with respect to classification algorithm $\A$ and training set $T$, if for all $\mathbf{x} \in X^{m}$ and $y \in Y$, we have 
\begin{displaymath}
    p_{\A(T)}(y \mid \mathbf{x}) \le e^{\epsilon} p_{\A(T \setminus \{t\})}(y \mid \mathbf{x})
\end{displaymath}
and 
\begin{displaymath}
    p_{\A(T)}(y \mid \mathbf{x}) \ge e^{-\epsilon} p_{\A(T \setminus \{t\})}(y \mid \mathbf{x}).
\end{displaymath}
\end{definition}

That is, the target record $t$ is DTP with algorithm $\A$ and training set $T$ if the predicted conditional probability of any class label $y$ given any feature vector $\mathbf{x}$ does not change much when $t$ is removed from $T$. By definition, a record $t$ with low DTP only has a small influence on the output of the classifier $\A(T)$. Since a classifier's output is also influenced by other factors such as random initialization and unexpected records in the training set, from the adversary's perspective, the small influence by $t$ is indistinguishable from the influence by other uncertain factors. Therefore, records with low DTP are less vulnerable to membership inference attacks. 

Unlike differential privacy~\cite{dwork2006differential}, DTP is a local property related to the training set. Therefore, DTP is experimentally measurable given a target record $t$, a training set $T$, and a classification algorithm $\A$. We use the following definition of DTP metric to quantify the privacy risk of the target record $t$. 

\begin{definition}
(DTP Metric).
Given classification algorithm $\A$, training set $T$, and target record $t$, the \emph{differential training privacy metric $\DTP_{\A,T}(t)$} is the least $\epsilon$ such that $t$ is $\epsilon$-DTP with $\A$ and $T$.
\end{definition}

In practice, $\DTP_{\A,T}(t)$ is calculated as the maximum ratio between the predictions given by $\A(T)$ and $\A(T \setminus \{t\})$ for all $\mathbf{x} \in X^{m}$ and for all $y \in Y$. That is,
\begin{displaymath}
\DTP_{\A,T}(t) = \underset{
  \mathbf{x} \in X^{m},  y \in Y 
}\max \epsilon_t^{(\mathbf{x},y)},
\end{displaymath}
where
\begin{displaymath}
\epsilon_t^{(\mathbf{x},y)} = \max \left( \frac{p_{\A(T)}(y \mid \mathbf{x})}{p_{\A(T \setminus \{t\})}(y \mid \mathbf{x})},  \frac{p_{\A(T \setminus \{t\})}(y \mid \mathbf{x})}{p_{\A(T)}(y \mid \mathbf{x})}\right).
\end{displaymath}

$\DTP_{\A,T}(t)$ can be naively measured by brute force over all $\mathbf{x} \in X^{m}$ and all $y \in Y$. However, considering the potentially large size of $X^{m}$, this approach is neither practical nor efficient. Therefore, we propose pointwise differential training privacy (PDTP) as a relaxation of DTP.

\begin{definition} \label{def:pdtp}
(Pointwise Differential Training Privacy). 
A record $t \in T$ is \emph{$\epsilon$-pointwise differentially training private} ($\epsilon$-PDTP) with respect to classification algorithm $\A$ and training dataset $T$, if for all $y \in Y$, we have 
\begin{displaymath}
    p_{\A(T)}(y \mid \mathbf{x}^{(t)}) \le e^{\epsilon} p_{\A(T \setminus \{t\})}(y \mid \mathbf{x}^{(t)})
\end{displaymath}
and 
\begin{displaymath}
    p_{\A(T)}(y \mid \mathbf{x}^{(t)}) \ge e^{-\epsilon} p_{A(T \setminus \{t\})}(y \mid \mathbf{x}^{(t)}).
\end{displaymath}
\end{definition}

Similarly we propose the following definition for the metric $\PDTP(\A, T)$:
\begin{definition}
(PDTP Metric).
Given classification algorithm $\A$, training set $T$, and target record $t \in T$, the \emph{pointwise differential training privacy metric $\PDTP_{\A,T}(t)$} is the least $\epsilon$ such that $t$ is $\epsilon$-PDTP with $\A$ and $T$.
\end{definition}

PDTP is a relaxation of DTP which bounds the change of the classifier's response on a single query $\mathbf{x}^{(t)}$, when $t$ is removed from the training set. Because of this, PDTP can be efficiently calculated given any classification algorithm $\A$, training set $T$, and target record $t$ by training an alternative classifier $\A(T \setminus \{t\})$. This process is similar to the leave-one-out (LOO) cross-validation technique used in machine learning~\cite{fukunaga1989leave}. Since LOO is a core technique for evaluating a machine learning model, there is considerable experience with both learning algorithms for which its calculation is easier and optimizations to improve this performance. 

The measurement $\PDTP_{\A,T}(t)$ is useful for two different reasons: First,~$\PDTP_{\A,T}(t)$ is a lower bound of $\DTP_{\A,T}(t)$. When $\DTP_{\A,T}(t)$ cannot be efficiently calculated, data owners can use $\PDTP_{\A,T}(t)$ as an optimistic estimation of a classifier's privacy risk. If a target record 
$t$ has high $\PDTP$ with $\A$ and $T$, releasing the classifier $\A(T)$ brings high privacy risks for $t$. Therefore, PDTP can be used as an indicator of membership privacy risk. Second, PDTP reflects the performance of a direct membership attack. When the adversary uses $c(\mathbf{x}^{(t)})$ to infer the membership of $t$, it is sufficient to bound the change of $c(\mathbf{x}^{(t)})$ when $t$ is removed from the training set of $c$. Since we assume direct membership attacks have better performance than indirect ones, $\PDTP_{\A,T}(t)$ is a good estimation of membership privacy risk of $t$. In Section~\ref{sec:case_studies}, we demonstrate the usefulness of PDTP measurements by showing their correlations with the performance of different types of direct membership attacks. 

%% file: case_studies.tex
\section{Case Studies}
\label{sec:case_studies}

DTP measures the sensitivity of a target record $t$ on the predictions of a classifier $\A(T)$. Intuitively, the larger $t$'s influence on the predictions of $\A(T)$ is, the more these predictions leak about $t$. This, in turn, makes $t$ more vulnerable to membership inference attacks. However, to use DTP to calculate the membership risk, we still need to answer the following: \emph{How do we use PDTP to estimate the risk of membership attacks? How accurate are these estimations? What values of PDTP indicate a potential privacy risk?}

In this section, we answer these three questions through a series of experiments on direct membership attacks. 

To demonstrate $\PDTP$'s effectiveness in measuring risks of membership attacks, we study the performance of three types of direct membership attacks on different datasets and classification algorithms. We find that, when a membership inference attack is effective, i.e., the attack accuracy is greater than $0.7$, $\PDTP_{\A,T}(t)$ is highly correlated with the attack's accuracy on $t$, and the correlation is higher for attacks with higher accuracy. 

To identify high-risk records, we use the {\em DTP-$1$ hypothesis}: if a classifier has a DTP value above $1$, it should not be published. Since PDTP is a lower bound for DTP, we use PDTP measurements to identify records that do not satisfy DTP-$1$ criterion and demonstrate effective membership attacks on these records. 

\subsection{Datasets}

We first introduce datasets used in the experiments.

\paragraphbe{UCI Adult Dataset (Adult).} The dataset~\cite{kohavi1996scaling} contains 48,842 records extracted from the 1994 Census Database. Each record has 14 attributes with demographic information such as age, gender, and education. The class attribute is the income class of the individual: $>50$K or $<=50$K. The classification task is to predict an individual's income class based on his demographic information. We use all the features except the final weight (fnlwgt) attribute, which is a weight on the Current Population Survey (CPS) file used for accurate populations estimates. We randomly sample 2,000 records as candidate set $D$, and 1,000 records out of $D$ as training set. 

\paragraphbe{Purchase Dataset (Purchase).} Similar to the purchase dataset in \cite{shokri2016membership}, we construct a dataset containing user's purchase history based on Kaggle's ``acquire valued shoppers'' challenge. The original contains the user's transaction histories, including product category, product brand, purchase quantity, purchase amount, etc. We pre-process the dataset by constructing one record for each customer. We use the product category attribute in the original dataset to create 836 binary feature attributes. Each feature attribute is a product category (e.g., sparkling water), and the value of the attribute is true if and only if the corresponding customer has purchased this product in the past year. We cluster the dataset into 10 clusters using k-means based on Weka's implementation~\cite{arthur2007k}. Each cluster represents a type of consumer buying behavior. We use the cluster index of each record as its class label. The classification task is to predict a consumer's buying behavior based on products he has purchased. We randomly sample 2,000 records as candidate set $D$, and 1,000 records out of $D$ as training set. 

\subsection{Machine Learning Models}

We study the performance of membership attack and PDTP measurements on three different machine learning models. 

\paragraphbe{Neural Networks (NN).} We build a fully connected neural network with one hidden layer of 64 units and a \texttt{LogSoftMax} layer. We use \texttt{Tanh} as the activation function and negative log-likelihood criterion as the classification criterion. We use a learning rate of 0.01 for both datasets. 
The maximum epoch of training is set to 100 for the adult dataset and 30 for the purchase dataset. When preprocessing the adult dataset, we convert categorical attributes into binary attributes and normalize all the numerical attributes. 

\paragraphbe{Naive Bayes Classifiers (NB).} We build a Naive Bayes classifier~\cite{rish2001empirical} to predict the class label based on Bayes Theorem under the assumption of conditional independence. We use Laplace smoothing~\cite{chen1996empirical} to smooth the categorical attributes in the dataset. 

\paragraphbe{Logistic Regressions (LR).} We build a logistic regression model using Weka's implementation of Logistic model trees~\cite{landwehr2005logistic}. 

\paragraph{Binning the Predictions.} We limit the precision of the model outputs using data binning technique with a bin size of 0.01. Since the outputs of the classifiers are probabilities in the range of $\left[0,1\right]$, we divide this range into 100 bins of the same size. Instead of returning the original output of a classifier, we make the model return the center of the bin to which the original output belongs. For example, if a classifier predicts a class label to have 0 probability given a feature vector, the output of the classifier would be 0.005, which is the center of the first bin. This binning technique prevents models from leaking private information that does not significantly contribute to their accuracy. It also prevents PDTP measurements from getting unreasonably large due to close-to-zero denominators.

\subsection{Attacks}

Given a target classifier $c = \A(T)$ and a target record $t$, a membership attack distinguishes between the following hypotheses:
\[ H_0: \ \ t \notin T \qquad \qquad \text{ and } \qquad \qquad H_1: \ \ t \in T \ . \]


In all of the following attacks, we assume that the adversary gets the target classifier's prediction on the target record $c(\mathbf{x}^{(t)}) = \left(\p_{c}(y_1 \mid \mathbf{x}^{(t)}), \p_{c}(y_2 \mid \mathbf{x}^{(t)}), \dots, \p_{c}(y_k \mid \mathbf{x}^{(t)})\right)$, and tries to launch a membership attack on $t$ using this information. We also assume the adversary is powerful enough to know the size of training set $T$ and has access to the candidate dataset $D$ and the classification algorithm $\A$ based on Kerckhoffs's principle. 

Let $q_i = \p_c(y_i \mid \mathbf{x}^{(t)})$ ($1 \leq i \leq k$). Here, $q_i$ represents the class probability for class label $y_i$ predicted by the target classifier given the features of the target record. Therefore, the vector $\mathbf{q} = c(\mathbf{x}^{(t)}) = (q_1, q_2, \dots, q_k)$ can be viewed as a probability distribution over all the possible class labels. 

\paragraphbe{Untargeted Attacks.} We reproduce the membership attack of~\cite{shokri2016membership} on a neural network model learned on the purchase dataset. This attack converts the membership inference problem into a classification task with two class labels: class label ``in'' represents hypothesis $H_1$ ($t \in T$), and class label ``out'' represents hypothesis $H_0$ ($t \notin T$). Concretely, the attack consists of two steps:

\paragraphe{Step 1: Training Shadow Classiffiers.} The adversary trains shadow classifiers to simulate the behavior of the target classifier $\A(T)$. First, he samples $M$ shadow datasets $T_1, T_2, \dots, T_M$ of the same size as the target dataset $T$. Then, he trains $M$ shadow classifiers $\A(T_1), \A(T_2), \dots, \A(T_M)$ using the same classification algorithm as the target classifier $\A(T)$. In experiments, we take $M=20$.

\paragraphe{Step 2: Building the Attack Classifier.} The adversary uses the shadow classifiers to label each record in the candidate dataset $D$ according to Algorithm~\ref{algo:untargeted}. The algorithm takes the shadow classifiers and the candidate dataset as input and outputs a dataset $D_{\rm attack}$, which serves as the training set for the attack classifier. 

\RestyleAlgo{boxruled}
\LinesNumbered
\begin{algorithm}[h]
\caption{Step 2 of Untargeted Attack}
\label{algo:untargeted}
\SetAlgoLined
    \KwIn{A set of shadow classifiers $\{\A(T_1), \A(T_2), \dots, \A(T_M)\}$, candidate set $D$}
    \KwOut{Training set of the attack classifier $D_{\rm attack}$}
    \BlankLine

    $D_{\rm attack} \leftarrow \emptyset$

    \For{$j = 1, 2, \dots, M$}{
        \For{$r \in D$} {
            $\mathbf{q}^{(r)} \leftarrow \left(\p_{\A(T_j)}(y_1 \mid \mathbf{x}^{(r)}),  \dots, \p_{\A(T_j)}(y_k \mid \mathbf{x}^{(r)})\right)$ \\
            $\mathbf{f}^{(r)} \leftarrow \left(\mathbf{x}^{(r)}, y^{(r)}, \mathbf{q}^{(r)}\right)$ \\

            \If{$r \in T_j$}{
                $D_{\rm attack} \leftarrow D_{\rm attack} \bigcup \left\{\left( \mathbf{f}^{(r)}, {\rm in} \right)\right\} $
            }
            \Else{
                $D_{\rm attack} \leftarrow D_{\rm attack} \bigcup \left\{\left( \mathbf{f}^{(r)}, {\rm out} \right)\right\} $
            }
        }
    }
    \Return $D_{\rm attack}$
\end{algorithm}

In experiments, the attack classifier is a fully connected neural network with one hidden layer of 32 units and a \texttt{LogSoftMax} layer. We use \texttt{ReLU} as activation function and negative log-likelihood criterion as classification criterion. We set the learning rate to 0.01 and the maximum epochs of training to 30 iterations. 

\paragraphe{Step 3: Launching the Attack.} Given a target record $t$, the adversary constructs a new feacture vector by concatenating the original feature vector $\mathbf{x}^{(t)}$, the original class label $y^{(t)}$, and the target model's prediction on the target record $c(\mathbf{x}^{(t)})$. That is,
\begin{equation*}
\mathbf{f}_{\rm attack}^{(t)} = \left(\mathbf{x}^{(t)},y^{(t)}, c(\mathbf{x}^{(t)}) \right).
\end{equation*}

The adversary queries the attack classifier with the new feature vector and gets a prediction consisting of two probabilities: $p_{\rm in}$ is the probability of class label ``in'', and $p_{\rm out}$ is the probability of class label ``out''. The advesary accepts hypothesis $H_1$ if, and only if, $p_{\rm in} > p_{\rm out}$. 

We call this type of attack an \emph{untargeted attack} because the attack classifier obtained from step 2 can be used to attack all the records in the candidate dataset. Therefore, when the adversary tries to attack multiple records, he only needs to run step 1 and step 2 once. Step 3 of the attack, which needs to be repeated on each targeted record, has much lower computational overhead. Although this attack is more efficient when the adversary wants to find out {\em any} vulnerable records, it has lower accuracy compared to some of the targeted attacks. 

\paragraphbe{Distance-Based Targeted Attacks.} When the adversary has a specific target record in mind, he can design attacks tuned to perform well only on the target record. We call this type of attacks \emph{targeted attacks}. In a distance-based targeted attack, an adversary uses shadow models to estimate the average predictions of classifiers satisfying hypothesis $H_0$, and those of classifiers satisfying hypothesis $H_1$. Then he calculates the distance between $c(\mathbf{x})$ and the two average predictions and accept the hypothesis under which the average predictions are closer to $c(\mathbf{x})$. Concretely:

\paragraphe{Step 1: Training Shadow Classifiers.}
Let $n = \left| T \right|$. The adversary uniformly samples $M$ datasets $T'_1, T'_2, \dots, T'_M$ of size $n-1$ from $D \setminus \{t\}$, and takes $T_j = T'_j \bigcup \{t\}$ for all $1 \leq j \leq M$. The adversary trains a pair of shadow classifiers $\A(T_j), \A(T'_j)$ for all $1 \leq j \leq M$, and gets their predictions on the target record:
\begin{equation*}
    \mathbf{p}^{\rm in}_j = \left( \p_{\A(T_j)}(y_1 \mid \mathbf{x}), \p_{\A(T_j)}(y_2 \mid \mathbf{x}), \dots, \p_{\A(T_j)}(y_k \mid \mathbf{x}) \right),
\end{equation*}
and
\begin{equation*}
    \mathbf{p}^{\rm out}_j = \left( \p_{\A(T'_j)}(y_1 \mid \mathbf{x}), \p_{\A(T'_j)}(y_2 \mid \mathbf{x}), \dots, \p_{\A(T'_j)}(y_k \mid \mathbf{x}) \right).
\end{equation*}
Take $\mean{\mathbf{p}}_{\rm in} = \frac{1}{M} \sum_{j=1}^{M} \mathbf{p}^{\rm in}_j$ and $\mean{\mathbf{p}}_{\rm out} = \frac{1}{M} \sum_{j=1}^{M} \mathbf{p}^{\rm out}_j$. Like the query result $\mathbf{q}$, $\mean{\mathbf{p}}_{\rm in}$ and $\mean{\mathbf{p}}_{\rm out}$ can be viewed as two probability distributions over all the possible class labels. 

\paragraphe{Step 2: Comparing KL-Divergence.}
The KL-Divergence~\cite{mackay2003information} between two distributions $P$ and $Q$ is defined to be
\begin{equation*}
    \kldiv{P}{Q} = \sum_{i} p_i \log \frac{p_i}{q_i}.
\end{equation*}
The adversary infers the membership of $t$ by comparing $\mathbf{q}$'s KL-Divergence to $\mean{\mathbf{p}}_{\rm in}$ and $\mean{\mathbf{p}}_{\rm out}$, and accepts hypothesis $H_1$ if, and only if, $\kldiv{\mathbf{q}}{\mean{\mathbf{p}}_{\rm out}} > \kldiv{\mathbf{q}}{\mean{\mathbf{p}}_{\rm in}}$.

In the experiment, we take $M = 5$. That is, for each target record $t$, we sample $5$ datasets and train $10$ shadow classifiers. $5$ of the shadow classifiers are trained with $t$ in the training set, and the other $5$ are trained without $t$ in the training set. 

\paragraphbe{Frequency-Based Targeted Attacks.} In a frequency-based targeted attack, the adversary trains the same shadow models as in the distance-based membership attack. However, instead of calculating the average of the predictions, the adversary counts the frequency that the predictions of classifiers satisfying hypothesis $H_0$ fall into the same bins as $c(\mathbf{x})$ as well as the frequency that the predictions of classifiers satisfying hypothesis $H_1$ fall into the same bins as $c(\mathbf{x})$. The adversary accepts the hypothesis under which predictions more often fall into the same bin as $c(\mathbf{x})$. 

The first step of a frequency-based targeted attack is the same as the distance-based targeted attack. The adversary trains $2m$ shadow classifiers, $m$ of which with $t$ in training set. In the second step, for $1 \leq i \leq k$, the adversary calculates $o_i^{in}$ as the number of shadow models that are trained with $t$ in training set and gives the same predicted probability on class label $y_i$ as the target dataset. Similarly, $o_i^{out}$ is calculated as the number of shadow models that are trained without $t$ in the training set and gives the same predicted probability on class label $y_i$ as the target dataset.

Finally the adversary estimates the following ratio:
\begin{equation*}
\frac{\Pr \left[ t \in T \mid c(\mathbf{x}) = \mathbf{q} \right]}{\Pr \left[ t \notin T \mid c(\mathbf{x}) = \mathbf{q} \right]} = \prod_{i = 1}^{k} \frac{o_i^{\rm in}}{o_i^{\rm out}}. 
\end{equation*}
The adversary accepts hypothesis $H_1$ iff $\frac{\Pr \left[ t \in T \mid c(\mathbf{x}) = \mathbf{q} \right]}{\Pr \left[ t \notin T \mid c(\mathbf{x}) = \mathbf{q} \right]} > 1$. 

Like for the distance-based membership attack, we take $M=5$ in the experiment and train $10$ shadow models for each target record. 

\begin{table*}[t]
\centering
{ \small
\caption{\small Performance of Three Membership Attacks on NN-Purchase.}
\label{table:nn-purchase}
\begin{tabular}{|c|c|c|c|c|c|c|}
\hline
\textbf{Membership Attack} & \textbf{Attack Accuracy} & \textbf{Attack Precision} & \textbf{Attack Recall} & \textbf{Attack F1 Score} & \textbf{Correlation with PDTP} & \textbf{$p$-Value} \\ 
\hhline{|=|=|=|=|=|=|=|}
Untargeted Attack & 0.6680 & 0.6386 & 0.8500 & 0.7294 & 0.4864 & $2.89 \times 10^{-7}$ \\
\hline
Frequency-Based Attack & 0.6257 & 0.5933 & 0.8253 & 0.7174 & 0.5052 & $8.29 \times 10^{-8}$ \\
\hline
Distance-Based Attack & 0.8533 & 0.8470 & 0.9087 & 0.8768 & 0.7653 & $1.85 \times 10^{-20}$ \\
\hline
\end{tabular}
}
\vspace{-6pt}
\end{table*}

\begin{figure*}
    \begin{subfigure}[b]{0.65\columnwidth}
        \includegraphics[width=\columnwidth]{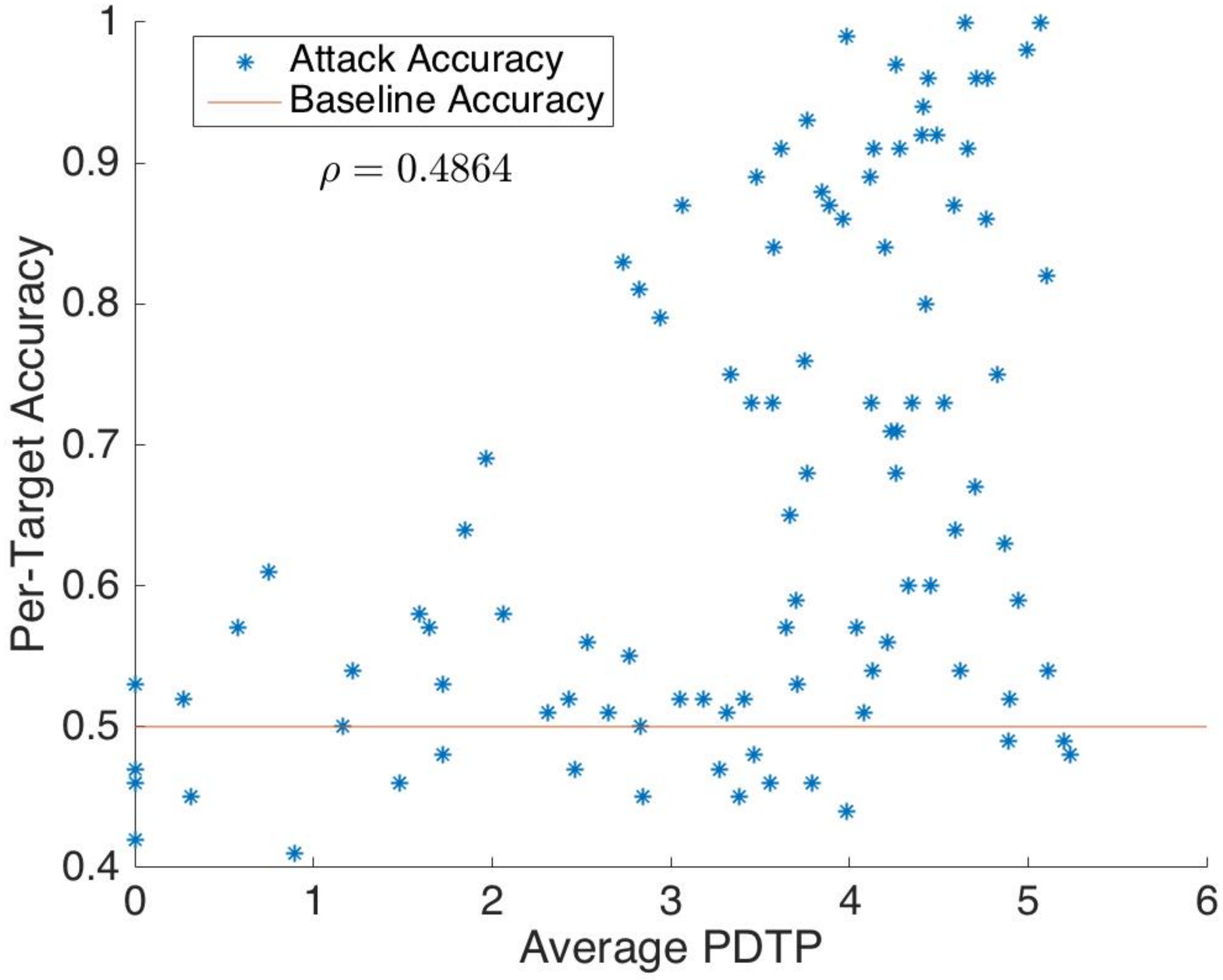}%
        \caption{\small Untargetted Membership Attack on NN-Purchase.}%
        \label{fig:untargeted_nn_purchase}%
    \end{subfigure}\hfill%
    \begin{subfigure}[b]{0.65\columnwidth}
        \includegraphics[width=\columnwidth]{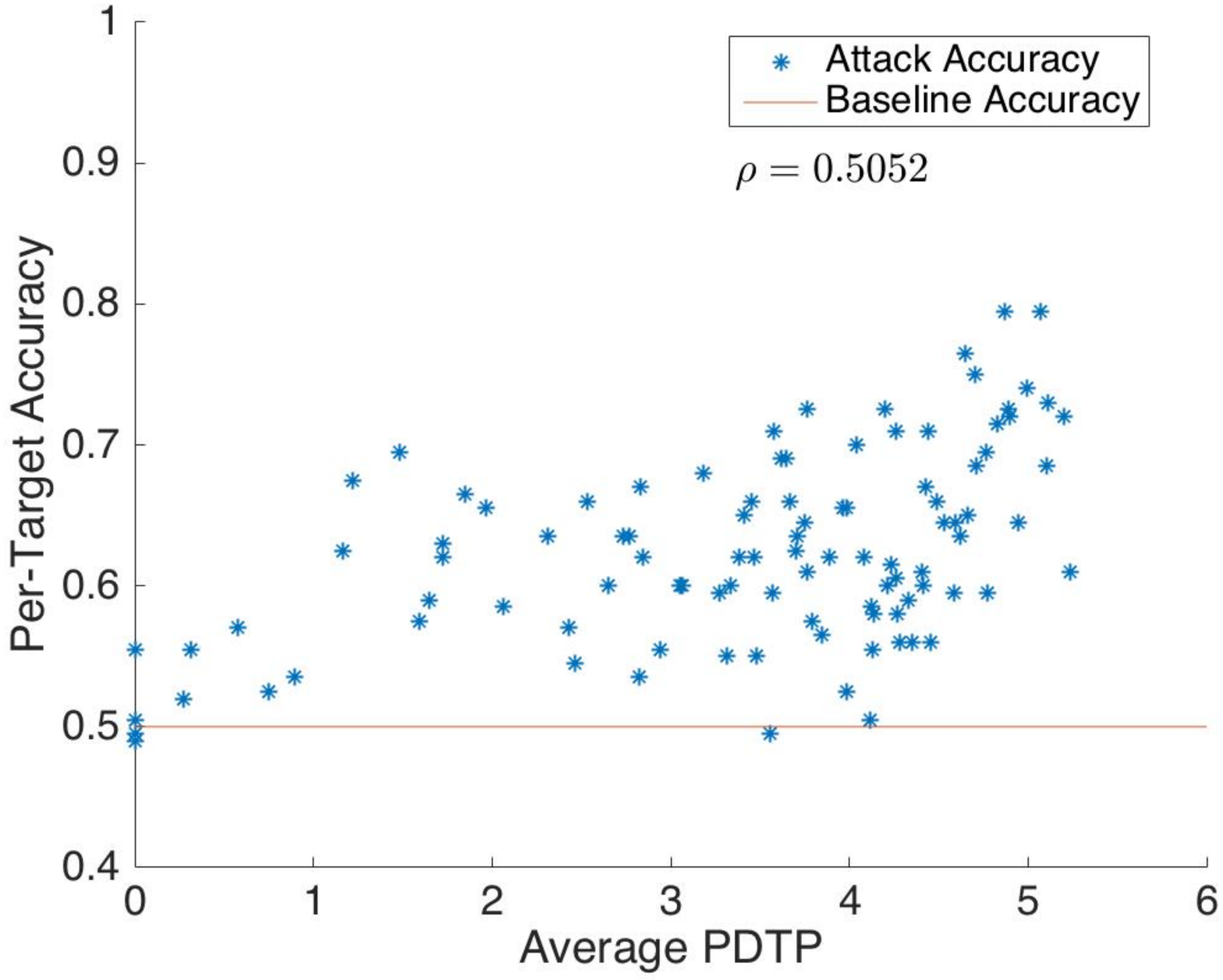}%
        \caption{\small  Frequency-Based Membership Attack on NN-Purchase.}%
        \label{fig:frequency_nn_purchase}%
    \end{subfigure}\hfill%
    \begin{subfigure}[b]{0.65\columnwidth}
        \includegraphics[width=\columnwidth]{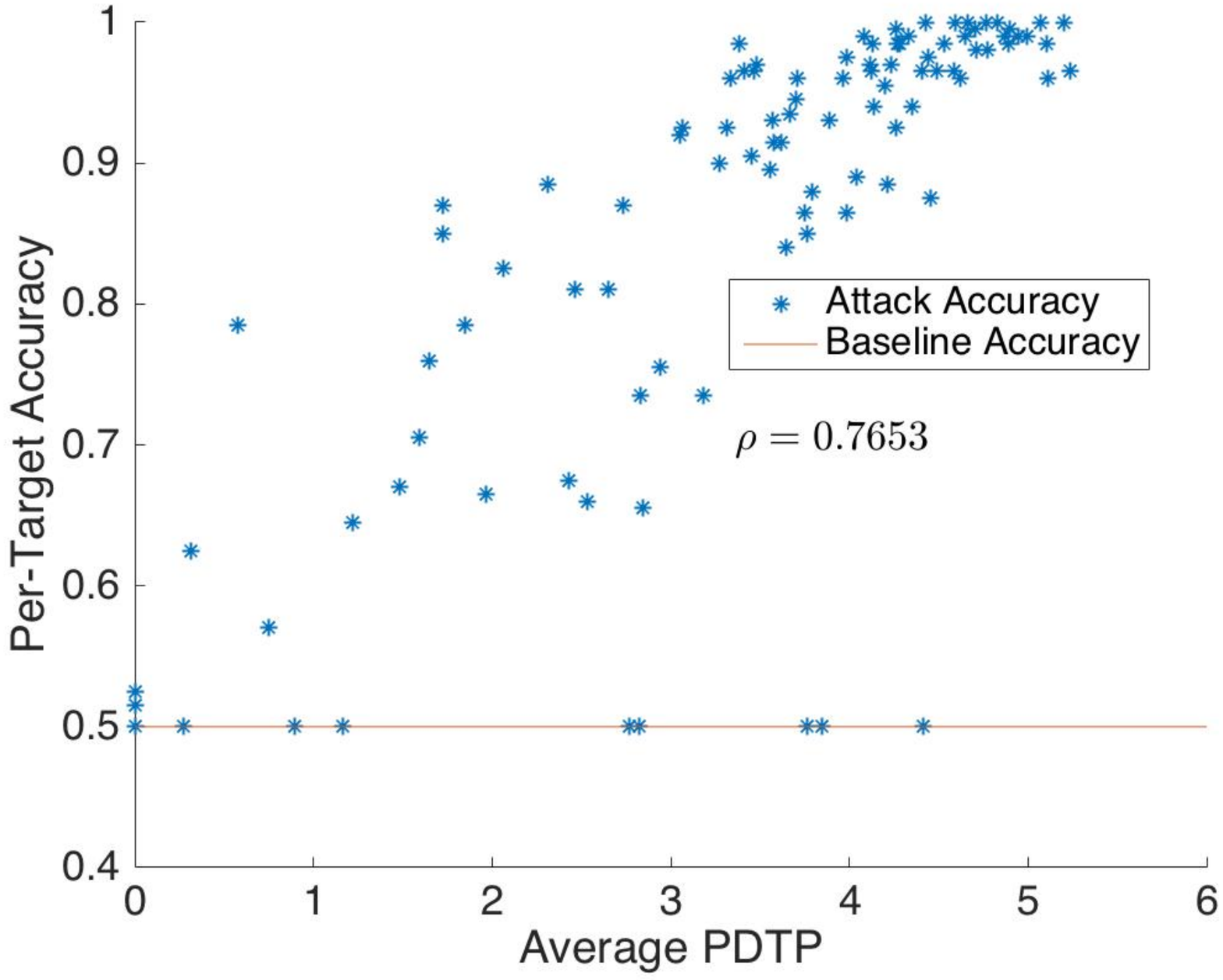}%
        \caption{\small Distance-Based Membership Attack on NN-Purchase.}%
        \label{fig:distance_nn_purchase}%
    \end{subfigure}\hfill%
    \caption{\small Membership Attacks on NN-Purchase. \vspace{-8pt}}%
\end{figure*}

\subsection{Evaluation Metrics}
\paragraphbe{Multiple Iterations of Attacks.} 
To evaluate their performance, we run 100 iterations of each membership attack. In each iteration, we partition the candidate set into two equal-sized parts: $D_1$ and $D_2$. First, we use $D_1$ as training set and $D_2$ as test set. A target classifier $\A(D_1)$ is trained on the training set and available for public queries. We randomly select 100 records out of $D$ as the adversary's target records. For each target record $t$, we run a targeted attack with $\A(D_1)$ as target classifier. In each membership attack, the goal of the adversary is to predict whether $t \in D_1$ by querying $\A(D_1)$. Then, we switch the role of $D_1$ and $D_2$, and use $D_2$ as the training set and $D_1$ as test set. We then repeat the membership attack with $\A(D_2)$ as target classifier. This process ensures that the target record occurs once in the training set and once in the test set in each iteration so that the baseline accuracy is always 0.5. Since we can only calculate PDTP of a record $t$ when $t$ is in the training set, we measure the PDTP of $t$ as $\PDTP_{\A,D_1}(t)$ if $t \in D_1$ and as $\PDTP_{\A,D_2}(t)$ if $t \in D_2$. Hence, in each iteration, we launch two membership attacks and get one PDTP measurement for each record in $D$ based on definition~\ref{def:pdtp}. We repeat this process for 100 iterations. Ideally, one should calculate a record's average PDTP over 100 PDTP measurements. However, PDTP measurements over multiple datasets contain redundant information. As shown in figure~\ref{fig:correlation}, average PDTP taken over 10 measurements has approximately the same correlation with performance of membership inference attacks as average PDTP taken over 100 measurements. Therefore, to save time, we only take PDTP measurements in the first $10$ iterations. For each target record, we use the average of its $10$ PDTP measurements to estimate its overall membership risk.

\paragraphbe{Per-Target Attack Accuracy.} 
We want to analyze the privacy risk of each record in $D$ separately. That is, instead of looking at the membership attack accuracy on each training set, we are interested in the overall attack accuracy on a single record over the 200 membership attacks. Therefore, we propose \emph{per-target attack accuracy} on a record $t$ as the adversary's proportion of correct membership inference on $t$ over all the attacks performed on $t$. For example, in the experiment, we launch 200 membership attacks on each record in $D$. Therefore, the per-target attack accuracy of a record is the number of correct membership inferences on that record divided by 200. 

\begin{figure}
    \includegraphics[width=0.65\columnwidth]{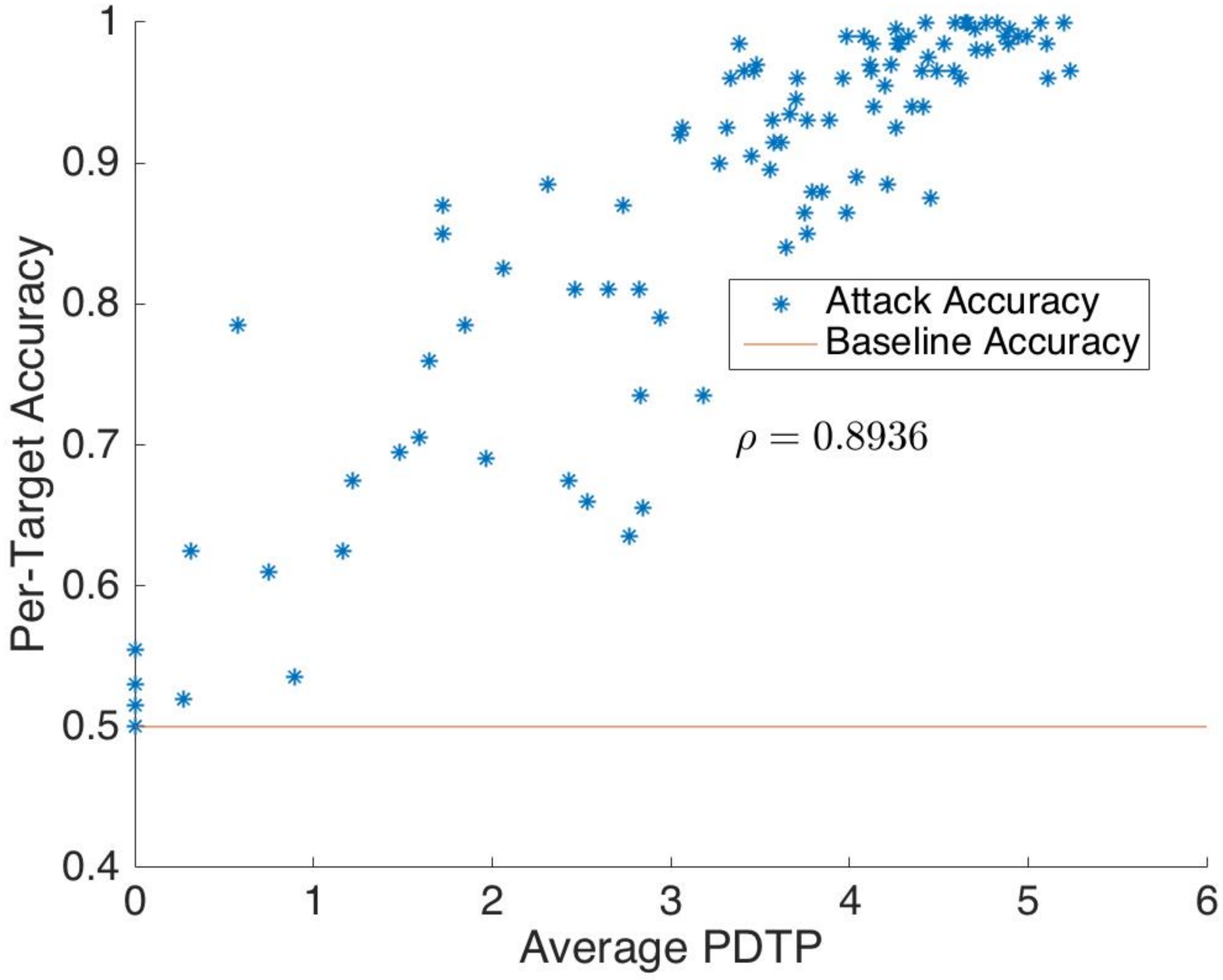}%
    \vspace{-2pt}
    \caption{\small Maximum Per-Target Accuracy of Three Membership Attacks on NN-Purchase. \vspace{-8pt}}%
    \label{fig:max_nn_purchase}%
\end{figure}%

\subsection{Results}
\label{subsec:result}

\paragraphbe{Comparison of Different Attacks.}
First, we compare the performance of three membership attacks on neural networks trained on the purchase dataset. We train 200 neural network models over different training sets sampled from the same candidate set and use them as the target classifiers for membership attacks. All the target classifiers are overfitted to their training sets. The average training accuracy of all the target classifiers is 1, and the average test accuracy of all the target classifiers is 0.6434.

Figures~\ref{fig:untargeted_nn_purchase},~\ref{fig:frequency_nn_purchase}, and~\ref{fig:distance_nn_purchase} show the per-target accuracy and average PDTP of each target record under three types membership attacks. Each point represents one target record in the candidate set. The horizontal axis is the average PDTP measurement of that target over $10$ iterations of PDTP measurements. The vertical axis is the per-target accuracy of that target over 200 repetitions of membership attacks. A point's position on the horizontal axis shows its membership privacy risk estimated by PDTP. According to PDTP measurements, points on the right part of the figures have higher membership privacy risks compared to points on the left part of the figures. A point's position on the vertical axis shows its actual membership privacy risk under a given membership attack. Points on the top part of the figures are more vulnerable to the attack because the attack has higher accuracy on these records. 

For each attack, we calculate the Pearson correlation coefficient between average PDTP and the per-target attack accuracy. We also calculate the p-value for testing the hypothesis of no correlation against the alternative hypothesis that there is a correlation between average PDTP and per-target attack accuracy. Table~\ref{table:nn-purchase} shows the performance of each membership attack and their correlation coefficients with average PDTP. The performance of all three attacks has statistically significant correlations with the average PDTP. Figure~\ref{table:nn-purchase} shows the accuracy of three types of membership attacks and their correlations with PDTP ($\rho$). We observe that attacks with higher accuracy also have higher correlation PDTP measurements. This correlation demonstrates PDTP's ability to identify potential membership privacy risks effectively. 

Overall, distance-based targeted attacks have the highest accuracy. They outperform the untargeted attacks in the previous work~\cite{shokri2016membership} by approximately $19\%$. However, some records are more vulnerable to some types of membership attacks. For example, one record in the purchase dataset is immune to distance-based membership attacks which only achieve baseline accuracy, whereas the untargeted attack achieves accuracy of 0.94. This example demonstrates the insufficiency of estimating privacy risks based on one type of attack. Even a strong attack may fail to identify some of the vulnerabilities that can be used by other attacks. Figure~\ref{fig:max_nn_purchase} shows the maximum per-target accuracy among three membership attacks. The Pearson correlation coefficient between the maximum accuracy and PDTP measurements is 0.8936. This is very strong correlation. Among all records with PDTP greater than 1, 84.62\% of them have maximum per-target accuracy higher than 0.8, and all of them have maximum per-target accuracy higher than 0.6. This result supports the DTP-1 hypothesis that classifiers with DTP above 1 should not be published.

\begin{figure}
        \includegraphics[width=0.65\columnwidth]{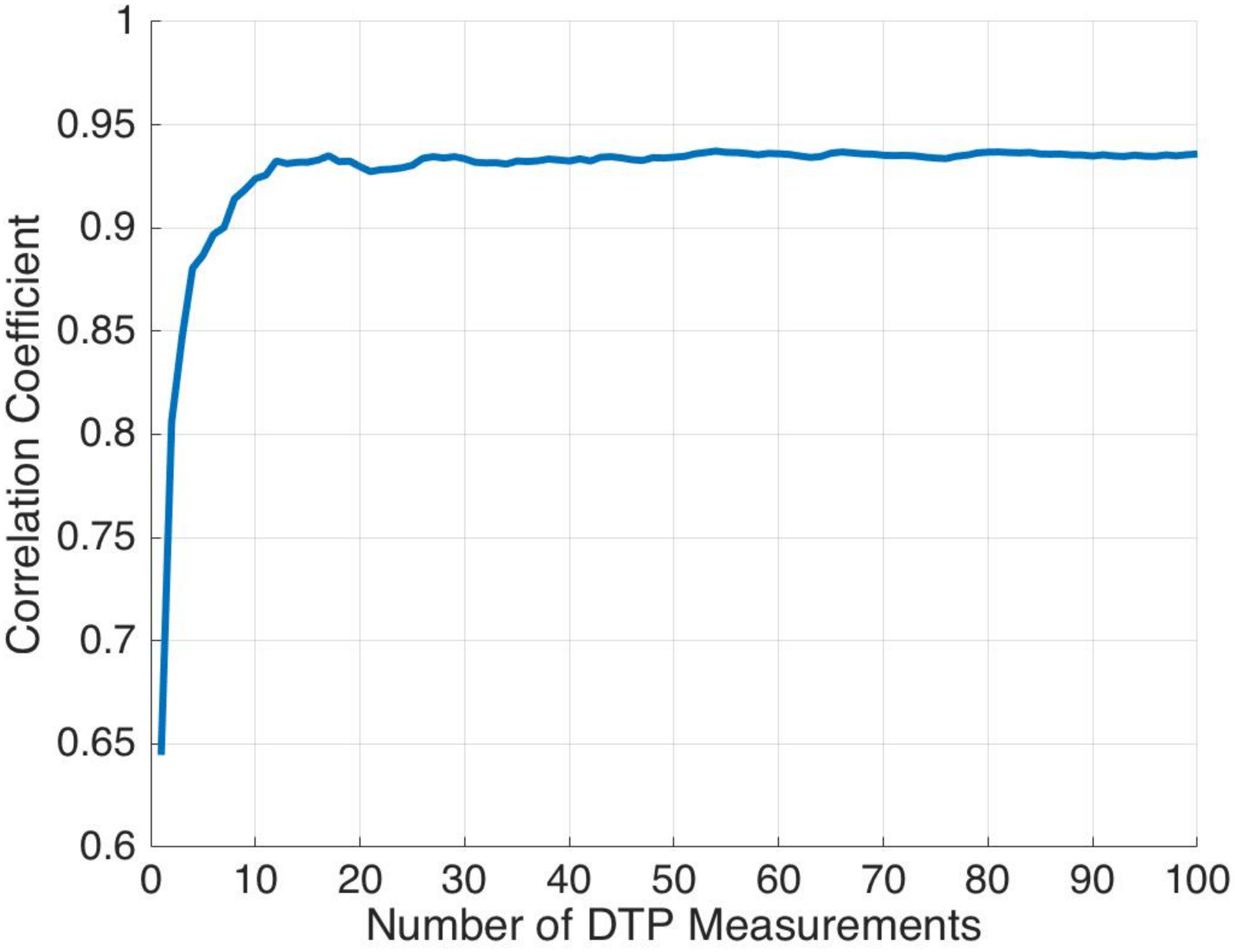}%
        \vspace{-3pt}
        \caption{\small Correlation between Average PDTP and Membership Attack Accuracy. \vspace{-8pt}}%
        \label{fig:correlation}%
\end{figure}

\paragraphbe{Privacy Risks of Different Models and Datasets.}
To compare the privacy risks of different datasets and classifiers, we use PDTP measurements on NN, NB, and LR classifiers learned on the adult and purchase datasets. We use distance-based membership attack because it has higher overall accuracy. Table~\ref{table:distance_attack} shows the performance of each target model, the per-target accuracy of membership attacks, and its correlation with PDTP measurements. 

\begin{table*}[t]
\centering
{ \small
\caption{\small Performance of Distance-Based attack on Different Target Models.}
\label{table:distance_attack}
\begin{tabular}{|c|c|c|c|c|c|c|c|c|c|}
\hline
\textbf{Target} & \textbf{Training} & \textbf{Test} & \textbf{Attack} & \textbf{Attack} & \textbf{Attack} & \textbf{Attack} & \textbf{Average} & \textbf{Correlation} & \textbf{$p$-Value} \\
\textbf{Model} & \textbf{Accuracy} & \textbf{Accuracy} & \textbf{Accuracy} & \textbf{Precision} & \textbf{Recall} & \textbf{F1 Score} & \textbf{PDTP} & \textbf{with PDTP} & \textbf{for Correlation} \\
\hhline{|=|=|=|=|=|=|=|=|=|=|}
NN-Purchase & 1.0000 & 0.6434 & 0.8533 & 0.8470 & 0.9087  & 0.8768 & 3.4019 & 0.7653 & $1.85 \times 10^{-20}$\\
\hline
NB-Purchase & 0.8641 & 0.7204 & 0.5958 & 0.6945 & 0.4038 & 0.5107 & 0.9027 & 0.9239  & $1.16 \times 10^{-42}$ \\
\hline
LR-Purchase & 1.0000 & 0.6241 & 0.7888 & 0.7314 & 0.9187 & 0.8144 & 2.8917  & 0.8138 & $7.78 \times 10^{-25} $\\
\hline
NN-Adult & 0.8555 & 0.8566 & 0.5340 & 0.5311 & 0.4402 & 0.4356 & 0.5847 & 0.4588 & $1.57 \times 10^{-6}$ \\
\hline
NB-Adult & 0.8453 & 0.8410 & 0.5128 & 0.5876 & 0.1027 & 0.1748 & 0.0299 & 0.5166 & $3.76 \times 10^{-8}$ \\
\hline
LR-Adult & 0.8711 & 0.8536 & 0.5134 & 0.5130 & 0.3818 & 0.4378 & 0.1460 & -0.0008 & $0.9343$ \\
\hline
\end{tabular}
}
\vspace{-6pt}
\end{table*}

The results of membership attacks on NB and LR models trained on the purchase dataset are shown in Figure~\ref{fig:distance_nb_purchase} and Figure~\ref{fig:distance_lr_purchase}. Both of the two classifiers have records with PDTP higher than 1, and most of these records are vulnerable to distance-based membership attacks. The accuracy of the attacks is highly correlated with PDTP measurements. 

\begin{figure*}
    \begin{subfigure}[b]{0.65\columnwidth}
    \includegraphics[width=\columnwidth]{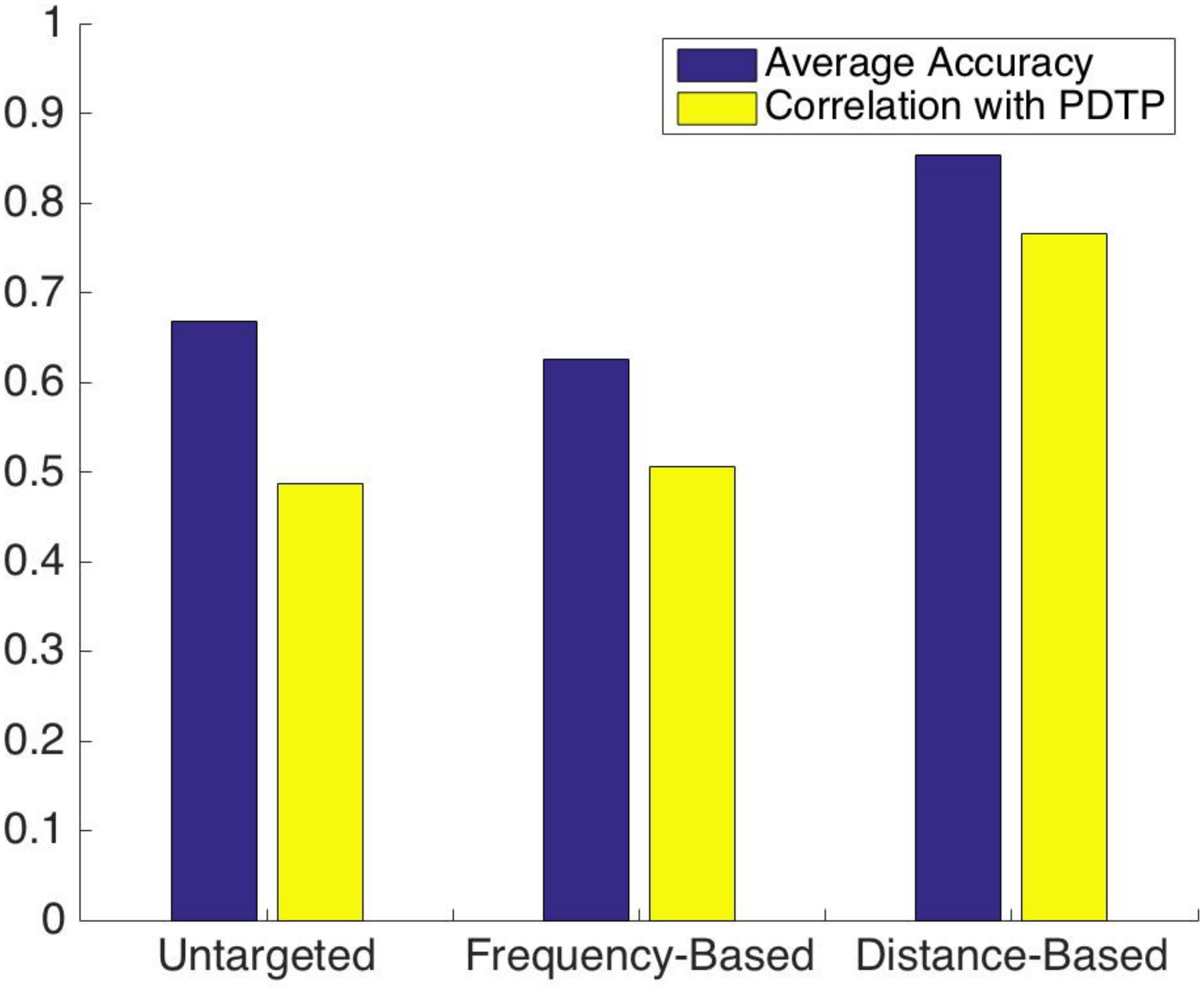}%
    \caption{Comparison between Three Membership Attacks on NN-Purchase10}%
    \label{fig:comparison}%
    \end{subfigure}\hfill%
    \begin{subfigure}[b]{0.65\columnwidth}
        \includegraphics[width=\columnwidth]{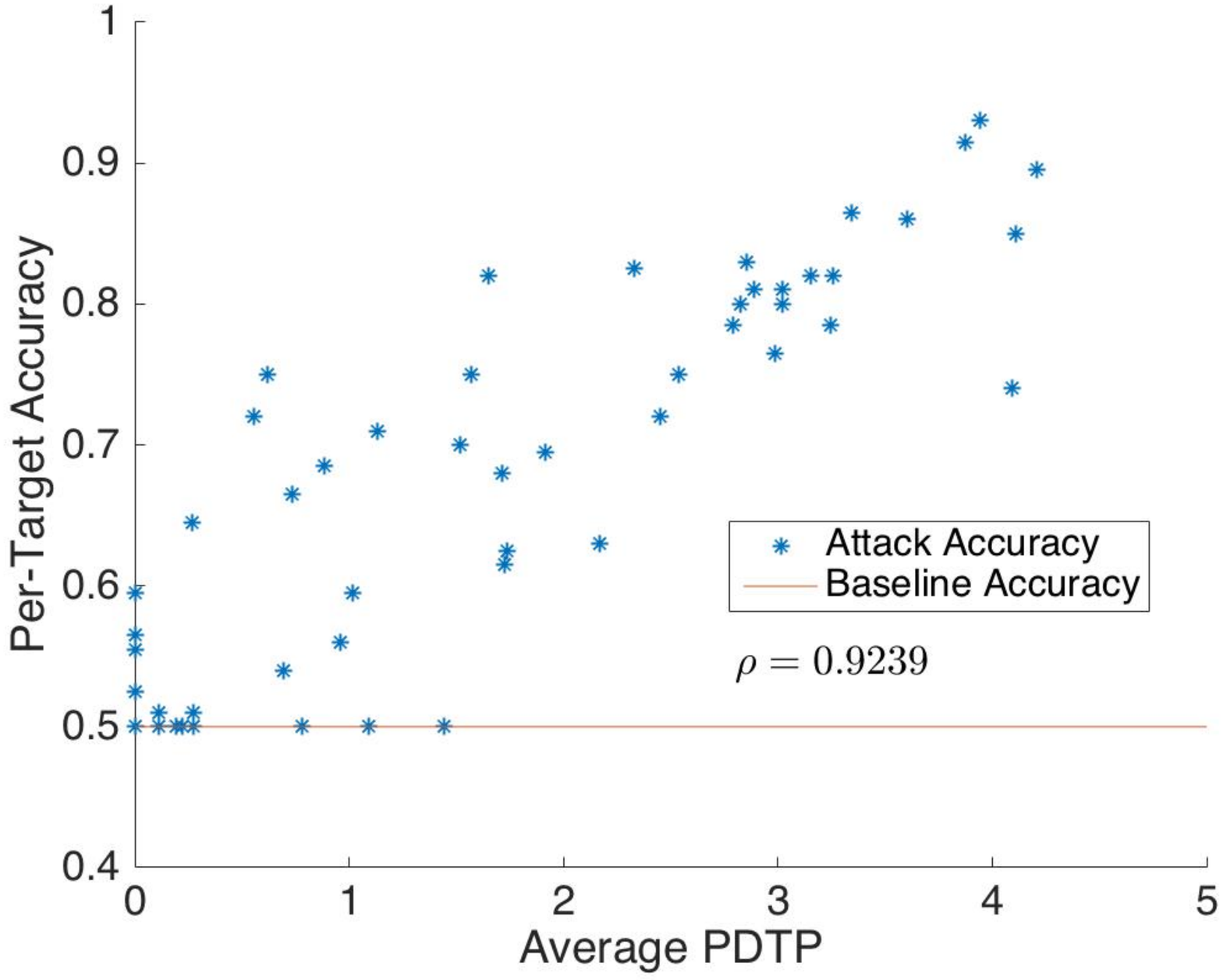}%
        \caption{Distance-Based Membership Attack on NB-Purchase}%
        \label{fig:distance_nb_purchase}%
    \end{subfigure}\hfill%
    \begin{subfigure}[b]{0.65\columnwidth}
         \includegraphics[width=\columnwidth]{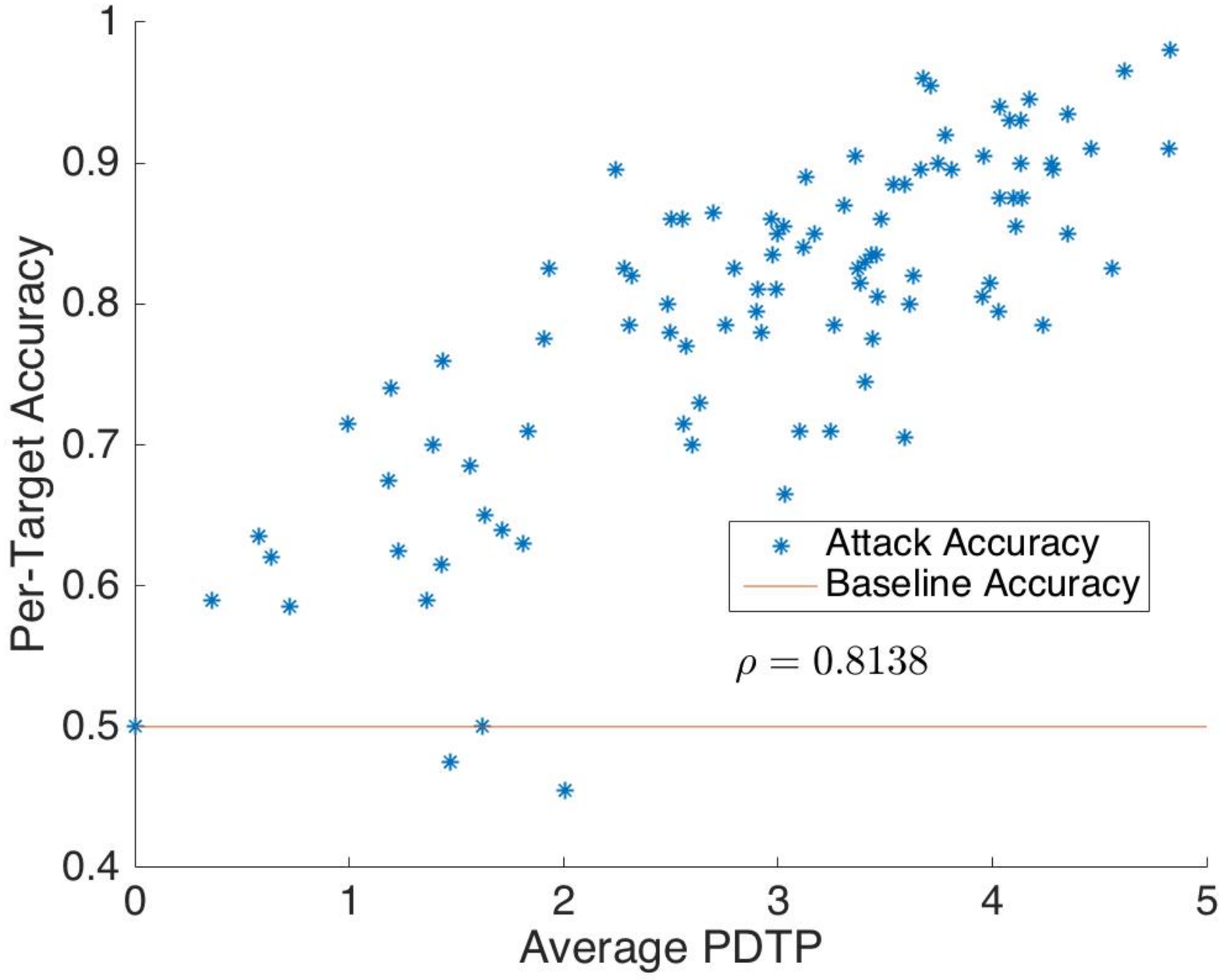}%
         \caption{Distance-Based Membership Attack on LR-Purchase}%
         \label{fig:distance_lr_purchase}%
    \end{subfigure}\hfill%
\caption{\small Membership Attacks on classifiers learned on purchase dataset}%
\end{figure*}

The results of membership attacks on NN, NB, and LR models trained on the adult dataset are shown in Figure~\ref{fig:distance_nn_adult}, Figure~\ref{fig:distance_nb_adult}, and Figure~\ref{fig:distance_lr_adult}. Unlike the purchase dataset, the adult dataset has fewer classes and features which help improve the generalizability of models learned on this dataset. The training and test accuracy reflects good generalizability of all three models learned on the adult dataset. The distance-based membership attacks also have worse performance on the adult dataset, indicating better membership privacy. However, even if the average PDTP measurement is relatively low for all three models, the PDTP for some records is greater than $1$ with the neural network model learned on the adult dataset indicating high membership privacy risk. This risk is also reflected by the high per-target accuracy of distance-based membership attack on some of the records with high PDTP. Therefore, good generalizability is not always sufficient for protecting membership privacy. It is possible that a model is not overfitted on the training set, but still captures some private information of some records in the training set. However, this privacy risk can be discovered by measuring PDTP for each record in the training set. 

For NB and LR models trained on the adult dataset, we did not find any records with PDTP greater than 1, and the per-target attack accuracy of the distance-based attack is smaller than $70\%$ for all target records. This result shows that state-of-art membership inference attacks do not work well on these models, and the PDTP measurements do not indicate high privacy risk for any of the records. The correlations between PDTP measurements and per-target attack accuracy is lower compared to attacks with better performance.

\paragraphbe{Multiple PDTP Measurements.}
In the previous experiments, for each target $t$, we use the average of 10 PDTP measurements on $t$ to estimate the PDTP of $t$ over all the target classifiers. To study how the number of PDTP measurements influence our estimation of the membership privacy risk of $t$, we take the Naive Bayes classifier trained on the purchase dataset as the target model perform a distance-based membership attack. We gradually increase the number of PDTP measurements from 1 to 100 and calculate the average PDTP's correlation with the accuracy of membership attacks. Figure~\ref{fig:correlation} shows that the correlation between average PDTP and accuracy of membership attacks increases as we increase the number of PDTP measurements. The correlation coefficient stabilizes after around 10 measurements. 

\begin{figure*}
    \begin{subfigure}[b]{0.65\columnwidth}
        \includegraphics[width=\columnwidth]{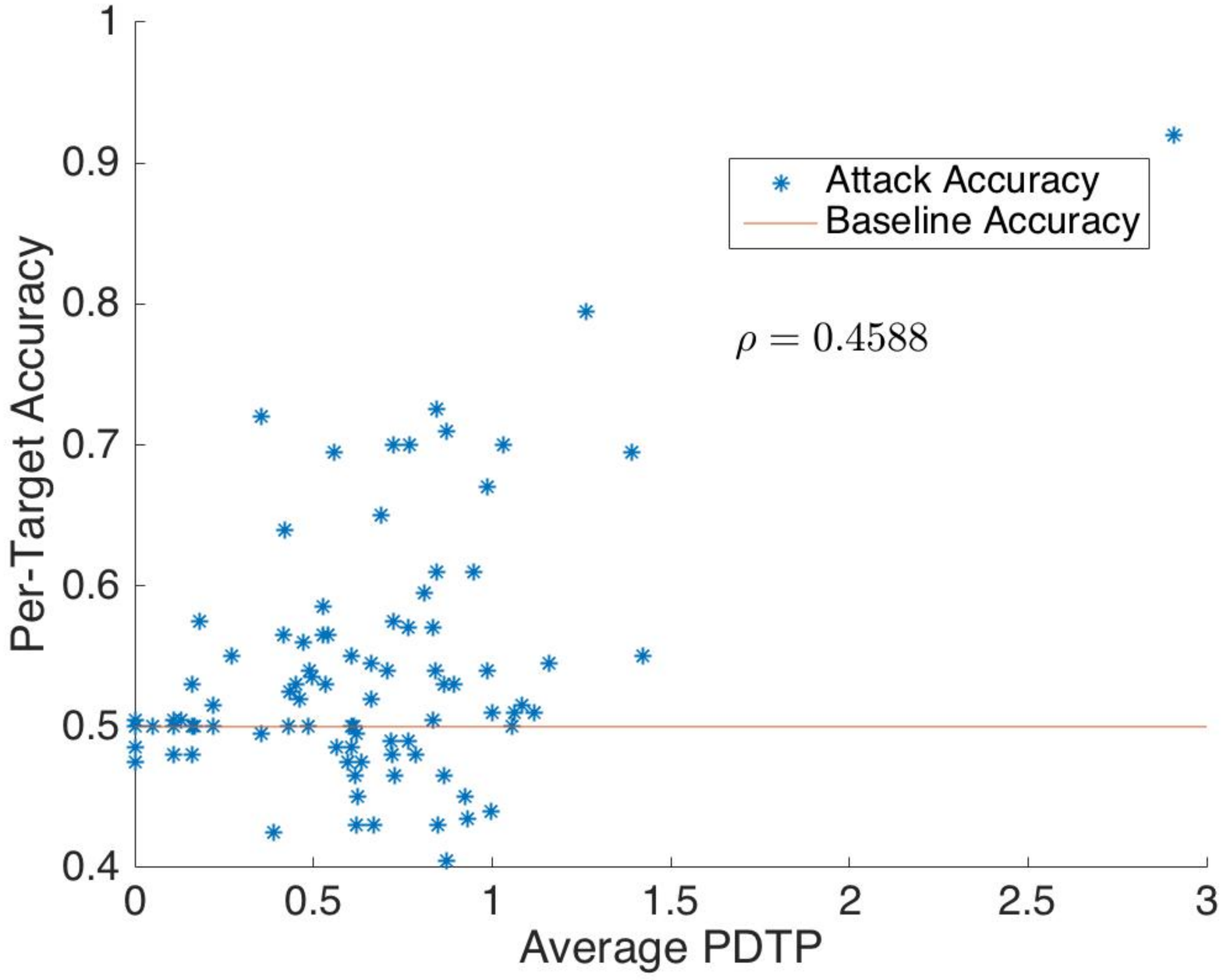}%
        \caption{\small Distance-Based Membership Attack on NN-Adult.}%
        \label{fig:distance_nn_adult}%
    \end{subfigure}\hfill%
    \begin{subfigure}[b]{0.65\columnwidth}
        \includegraphics[width=\columnwidth]{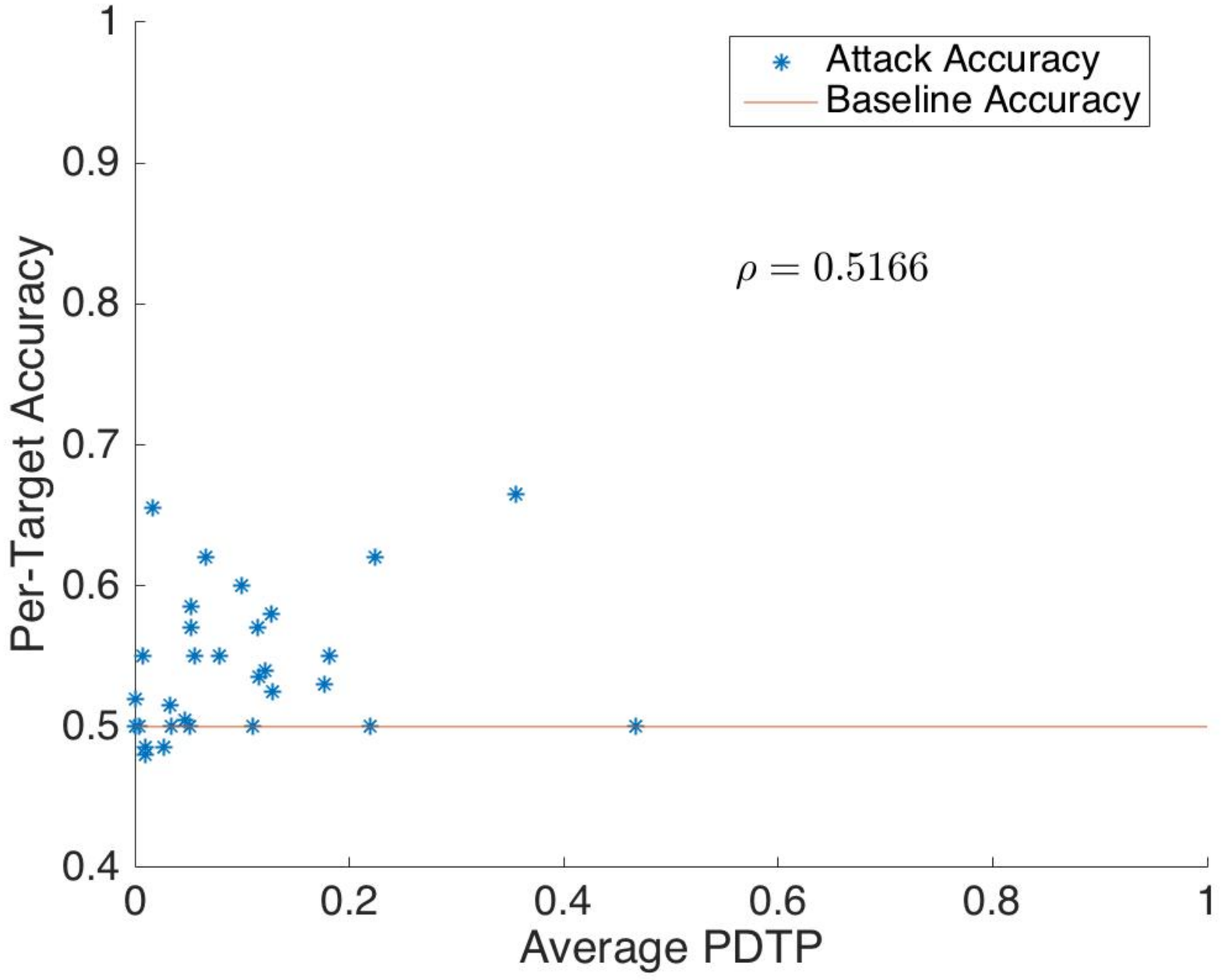}%
        \caption{\small Distance-Based Membership Attack on NB-Adult}%
        \label{fig:distance_nb_adult}%
    \end{subfigure}\hfill%
    \begin{subfigure}[b]{0.65\columnwidth}
         \includegraphics[width=\columnwidth]{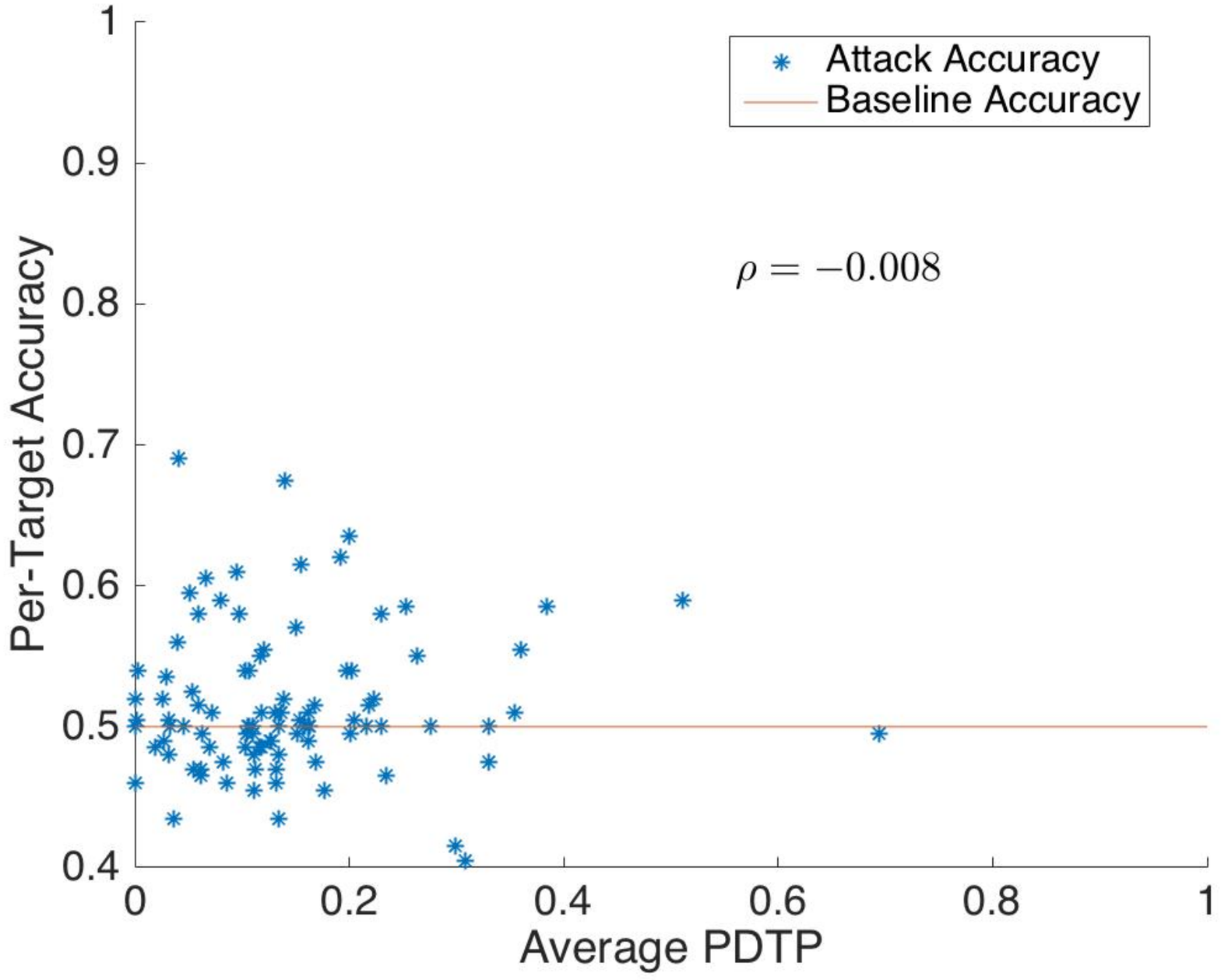}%
         \caption{\small Distance-Based Membership Attack on LR-Adult.}%
         \label{fig:distance_lr_adult}%
    \end{subfigure}\hfill%
\caption{\small Membership Attacks on classifiers learned on adult dataset. \vspace{-8pt}}%
\end{figure*}

%% file: dtp_measure.tex

\section{Protections against Indirect Membership Attacks}
\label{sec:measure_dtp}

In this section, we investigate the risk of indirect membership attacks where the adversary queries the classifier for features other than the target record.

\subsection{Risk}

In the previous experiments, we assume that the best way of doing a membership attack is to launch a direct attack by querying the target record. However, is it possible that, for some classifier $c$, there exists a query $\mathbf{x} \neq \mathbf{x}^{(t)}$ so that $c(\mathbf{x})$ leaks more private information than $c(\mathbf{x}^{(t)})$? That is, can an indirect membership attack outperform any direct membership attacks? Although it is hard to design a good indirect membership attack for classifiers discussed in Section~\ref{sec:case_studies}, this risk of indirect membership attacks can be demonstrated with a carefully designed classifier that encodes membership information of one specific record. 

Let $c = \A(T)$ be a classifier learned on a training set $T$ with machine learning algorithm $\A$. Instead of releasing $c$, we construct a classifier $c^*$ as follows: 
\begin{equation*}
\begin{aligned}
	c^*(\mathbf{x}) = 
	\begin{cases}
	c(\mathbf{x}) \qquad & \mbox{if} \quad \mathbf{x} \neq \mathbf{0} \\
	\mathbf{1} \qquad & \mbox{if} \quad \mathbf{x} = \mathbf{0} \quad \mbox{and} \quad t \in T \\
	\mathbf{0} \qquad & \mbox{if} \quad \mathbf{x} = \mathbf{0} \quad \mbox{and} \quad t \notin T.
	\end{cases}
\end{aligned}
\end{equation*}

Assume $\mathbf{x}^{(t)} \neq \mathbf{0}$, apparently, an indirect membership attack with query $\mathbf{x} = \mathbf{0}$ gives more information about the target $t$ compared to a direct membership attack with query $\mathbf{x}^{(t)}$. This example shows that for some classifiers, indirect attacks can outperform direct attacks for some records. Therefore, a record can have high membership privacy risk even if its PDTP measurement is low. 

Clearly, $c^*$ is not representative of a real-life machine learning model, especially when the model is trained by the data owner who wants to protect against privacy leakage. However, to achieve a stronger privacy guarantee, we need to study the potential risk of indirect membership attacks and prevent models from leaking ``side channel'' information about records in their training sets.

\subsection{Training Stability}
\label{subsec:training_stable}
To protect against indirect membership attacks, we need a way of calculating $\DTP_{\A,T}(t)$ without the need of brute forcing the whole feature space $X^m$. Since we can already efficiently calculate $\PDTP_{\A,T}(t)$, a natural approach is to bound $\DTP_{\A,T}(t)$ based on $\PDTP_{\A,T}(t)$. We call this property training stability. 

\begin{definition}
(Training Stability)
A classification algorithm $\A $ is \textit{$\delta$-training stable} on dataset $T$ if there exists a constant $\delta > 1$, so that for all $t \in T$ with $p_{A(T \setminus \{t\})}(y^{(t)} \mid \mathbf{x}^{(t)}) > 0$ and  $p_{A(T)}(y^{(t)} \mid \mathbf{x}^{(t)}) > 0$, for all $\mathbf{x} \in X^{m}$, for all $y \in Y$, let
\begin{displaymath}
\gamma_t = \max(\delta,\frac{p_{\A(T)}(y^{(t)} \mid \mathbf{x}^{(t)})}{p_{A(T \setminus \{t\})}(y^{(t)} \mid \mathbf{x}^{(t)})}, \frac{p_{\A(T \setminus \{t\})}(y^{(t)} \mid \mathbf{x}^{(t)})}{p_{A(T)}(y^{(t)} \mid \mathbf{x}^{(t)})}),
\end{displaymath}
we have 
\begin{displaymath}
p_{\A(T)}(y \mid \mathbf{x}) \le \gamma_t p_{A(D \setminus \{t\})}(y \mid \mathbf{x}),
\end{displaymath}
and
\begin{displaymath}
p_{\A(T)}(y \mid \mathbf{x}) \ge \gamma_t^{-1}p_{A(D \setminus \{t\})}(y \mid \mathbf{x}).
\end{displaymath}
\end{definition}

Given an algorithm that is $\delta$-training stable on $T$, for all $t \in T$, the ratio between the predictions of two classifiers $\A(T)$ and $\A(T \setminus \{t\})$ is bounded either by the ratio between their predictions on the query $(\mathbf{x}^{(t)}, y^{(t)})$ or by a parameter $\delta$.

If an algorithm $\A$ is $\delta$-training stable on $T$, $\DTP_{\A, T}(t)$ can be calculated by measuring $\PDTP_{\A,T}(t)$, which is much more efficient.

\begin{theorem}
\label{theorem:stability}
If a record $t \in T$ is $\epsilon$-PDTP with classification algorithm $\A$ and dataset $T$, and $\A$ is $\delta$-training stable on $T$, we have $t$
is $\epsilon'$-DTP with $\A$ and $T$, where $\epsilon' = \max(\epsilon, \ln\delta)$.
\end{theorem}

\begin{proof}
See Appendix~\ref{appendix:proof}.
\end{proof}

On the one hand, $\delta$-training stability is a desirable property from a privacy perspective because it reduces the computational cost of estimating the influence of an individual record on the learned classifier. On the other hand, $\delta$-training stability is also a metaphor of learning in real life. For example, If a professor explains an example question in class, he expects the students to do well on similar questions in the exam. If the exam contains a question that is the same as the example question explained in class, most students are expected to answer it correctly. Similarly, suppose we have a classifier $\A(T \setminus \{t\})$ and an additional record $t$. By adding $t$ into the training dataset, we expect the classifier to have better performance on $t$ or records similar to $t$. This can be viewed as a metaphor of learning in real life. 

\subsection{Training Stability of Classifiers}
With the aforementioned intuitions in mind, we study the training stability of some commonly used classifiers. However, due to the complexity and variability of different machine learning algorithms, we cannot cover all well-known classifiers in this section. Table~\ref{table:training_stability} shows the training stability of some commonly used classifiers. 

\paragraphbe{Bayes Inference Classifiers.} 
For a Bayes inference classifier $\A(T)$, the prediction $p_{\A(T)}(y \mid \mathbf{x})$ is given by the conditional probability of class label $y$ given feature vector $\mathbf{x}$ in the training dataset $T$. 

\begin{proposition}
\label{prop:bi}
Bayes inference algorithm is $\delta$-training stable for $\delta = \frac{4}{3}$ on any training dataset.
\end{proposition}
\begin{proof}
See Appendix~\ref{appendix:proof}.
\end{proof}

\paragraphbe{Naive Bayes Classifiers.} 
Naive Bayes classifiers make predictions using Bayes theorem and assume conditional independence~\cite{rish2001empirical}. 

\begin{proposition}
\label{prop:nb}
Let $T$ be a training dataset with $m$ features and $n$ examples. Let $y_{\min}$ be the least supported class label in $T$.  Let $n_{y_{\min}}$ be the number of examples with class label $y_{\min}$.  Naive Bayes algorithm is $\delta$-training stable for 
\begin{displaymath}
\delta = \left(\frac{n_{y_{\min}}}{n_{y_{\min}} - 1}\right)^{m-1} \frac{n}{n - 1}.
\end{displaymath}
\end{proposition}
\begin{proof}
See Appendix~\ref{appendix:proof}.
\end{proof}

If $T$ is a large training dataset, there would be a large number of training examples with class label $y_{\min}$. Therefore, $\delta$ would be close to 1 for a large dataset $T$, and the maximum ratio between predictions of $\A(T)$ and $\A(T \setminus \{t\})$ is determined by the ratio between their predictions on the query $(\mathbf{x}^{(t)}, y^{(t)})$.

Naive Bayes classification is often used with Laplace smoothing~\cite{schutze2008introduction}. When conditional probabilities are estimated from the training dataset, a small constant is added to both the numerator and denominator to get a "smoothed" version of the prediction. The constant is determined by the number of possible values for each attribute. Suppose each attribute in $\mathbf{x}$ has at most $v$ possible values. When calculating the conditional probability of an attribute given the class label, the numerator is increased by 1, and the denominator is increased by $v$. Therefore, naive Bayes classification algorithm with Laplace smoothing is $\delta$-training stable with a slightly different $\delta$ compared to the original naive Bayes classification.

\begin{proposition}
Let $T$ be a training dataset with $m$ features and $n$ examples. Suppose each element in the feature vector has at most $v$ possible values. Let $y_{\min}$ be the least supported class label in $T$.  Let $n_{y_{\min}}$ be the number of examples with class label $y_{\min}$.  Naive Bayes with Laplace smoothing is $\delta$-training stable on $T$ for 
\begin{displaymath}
\delta = \left(\frac{n_{y_{\min}} + v}{n_{y_{\min}}}\right)^{m-1} \frac{n}{n - 1}.
\end{displaymath}
\end{proposition}

\paragraphbe{Linear Statistical Queries Classifiers.} 
Linear statistical queries (LSQ) classifiers are proposed as a generalization framework for naive Bayes, Bayesian network, and Markov models~\cite{roth1999learning}. 

Let $\chi : X^{m} \rightarrow \{0,1\}$ be a feature function that maps a feature vector into a binary value. This representation is useful for features depending on more than one element in $\mathbf{x}$ (for example, $\chi(\mathbf{x}) = 1$ iff $x_1 =1$ and $x_2 = 1$). A \emph{statistical query} $\hat{P}_{[\chi, y]}^T$ gives the probability of all the examples with feature $\chi(\mathbf{x}) = 1$ and class label $y$ in the training dataset. 

A linear statistical queries (LSQ) classifier is a linear discriminator over the feature space, with coefficients calculated by statistical queries. For the convenience of discussion, we review the following definition of LSQ classifier for binary classification:
\begin{definition}[Linear Statistical Queries classifier \cite{roth1999learning}]
Let $\mathcal{X}$ be a class of features. Let $f_{[\chi, y]}$ be a function that depends only on the values $\hat{P}_{[\chi,y]}^T$ for $\chi \in \mathcal{X}$. A \emph{linear statistical queries (LSQ) hypothesis} predicts $y \in \{0,1\}$ given $\mathbf{x} \in X^{m}$  when 
\begin{displaymath}
y = \underset{{y \in \{0,1\}}}{\arg\max} \sum_{\chi \in \mathcal{X}} f_{[\chi,y]}(\hat{P}_{[\chi,y]}^T)\chi(\mathbf{x}).
\end{displaymath}
\end{definition}

We define a family of \emph{log coefficient functions $\mathcal{F}_{\log}$} that contains all the functions $f_{[\chi,y]}$ that calculate the log of a probability or conditional probability in the training dataset. For example, suppose $\A$ is a naive Bayes classification algorithm with $m$ feature attributes. $\A$ can be written as an LSQ classification algorithm with $m$ features: $\chi_0 \equiv 1$, and $\chi_j = x_j$ for $1 \le j \le m$, where
\begin{displaymath}
f_{[\chi_0,y]}(\hat{P}_{[\chi_0,y]}^T) = \log \hat{P}_{[1,y]}^T,
\end{displaymath}
and for $1 \le j \le m$,
\begin{displaymath}
f_{[\chi_j,y]}(\hat{P}_{[\chi_j,y]}^T) = \log\rfrac{\hat{P}_{[x_j,y]}^T}{\hat{P}_{[1,y]}^T}.
\end{displaymath}
$f_{[\chi_0,y]}$ is a log function of the prior probability of $y$, and $f_{[\chi_j,y]}$ is a log function of the conditional probability of $x_j$ given $y$. Therefore, the coefficient functions for naive Bayes belong to the family of log coefficient functions. Similarly, log coefficient functions are also used in Bayes network and Markov model. 

When $f_{[\chi,y]} \in \mathcal{F}_{\log}$, the sum of $f_{[\chi,y]}$ is equivalent to the product of the corresponding probabilities. Therefore, in addition to returning the most likely label, an LSQ classifier $\A(T)$ is also a probabilistic classifier that returns the following predicted probability:
\begin{equation}
\label{eq:lsq}
\begin{aligned}
p_{\A(T)}(y \mid \mathbf{x}) &= e^{\underset{\chi \in \mathcal{X}}{\sum} f_{[\chi,y]}(\{\hat{P}_{[\chi,y]}^T\})\chi(\mathbf{x})} \\
&=\prod_{\chi \in \mathcal{X}} e^{f_{[\chi,y]}(\{\hat{P}_{[\chi,y]}^T\})\chi(\mathbf{x})}.
\end{aligned}
\end{equation}

Each term $e^{f_{[\chi,l]}(\{\hat{P}_{[\chi,l]}^T\})\chi(\mathbf{x})}$ in Equation ($\ref{eq:lsq}$) is equivalent to calculating a probability or a conditional probability in the training dataset $T$ using Bayes inference, therefore is $\frac{4}{3}$-training stable according to proposition~\ref{prop:bi}. Consequently, we have the following proposition for LSQ probabilistic classification algorithms:
\begin{proposition}
\label{prop:lsq}
If $\A$ is an LSQ probabilistic classification algorithm with $f_{[\chi,y]} \in \mathcal{F}_{\log}$ for all $\chi \in \mathcal{X}, y \in Y$, $\A$ is $\delta$-training stable on any training dataset with $\delta = (\frac{4}{3})^{|\mathcal{X}|}$. 
\end{proposition}

For naive Bayes classification algorithm, since each attribute is an independent feature, with $M-1$ feature attributes, $|\mathcal{X}|$ equals $M$ . Compared to proposition~\ref{prop:nb}, proposition~\ref{prop:lsq} provides a looser but more generalized bound on $\delta$-training stability for naive Bayes classification algorithm. This bound does not depend on the records in the training dataset $T$. For Bayesian network and Markov models, $|\mathcal{X}|$ equals to the layer of dependencies in the network. $|\mathcal{X}|$ gets larger when the network structure gets more complicated.

\begin{table*}
\centering
{ \small
\caption{\small Training Stability of Different Classifiers. \vspace{-3pt}}
\label{table:training_stability}
\begin{tabular}{|c|c|c|}
\hline
\textbf{Training Stable Classifiers} & \textbf{Training Stability Unknown} & \textbf{Non-Training Stable Classifiers} \\
\hhline{|=|=|=|}
Bayes Inference Classifiers, Naive Bayes Classifiers, & Support Vector Machines & k-Nearest Neighbors \\
Linear Statistical Queries  & Neural Networks &  \\
(e.g., Bayes Networks, Markov Models)  & Logistic Regressions  & \\
Random Decision Trees & ... &  \\
\hline
\end{tabular}
}
\vspace{-6pt}
\end{table*}

\paragraphbe{Decision Trees.}
Some decision trees, such as ID3 and C4.5, construct the structure of the tree by calculating information gain of each potential partition of attributes\cite{quinlan1986induction}. This approach makes achieving $\delta$-training stability difficult because when an example is removed from the training dataset, it is hard to predict its influence on the structure of the tree. For example, removing one example may change the splitting point with the highest information gain, so that the structure of $\A(T)$ and $\A(T \setminus \{t\})$ are completely different. 

However, when the structure of a decision tree is independent of its training dataset, its prediction is equivalent to the conditional probability of $y$ given a subset of attributes determined by the leaves of the tree. Therefore, a single decision tree with structure independent of its training dataset is $\frac{4}{3}$-training stable. 

A random decision tree classifier is a classifier constructed by aggregating $K$ randomly generated decision trees with structures independent of the training dataset\cite{fan2003random}. Random decision trees have better privacy properties because the structure of the trees do not leak private information about the training set. Previous work has shown that a large amount of noise is needed to make ID3 differentially private while it is more practical to achieve differential privacy for a random decision tree~\cite{jagannathan2009practical}.

If the predictions of the random decision tree classifier is aggregated in a way that preserves the training stability, the random decision tree classifier is also training stable.
\begin{proposition}
\label{prop:dt}
Let $\A_K$ be a random decision tree classification algorithm with $K$ randomly generated decision trees. Given a query $(\mathbf{x},y)$, let $\p_1(y \mid \mathbf{x}), \p_2(y \mid \mathbf{x}), \dots, \p_K(y \mid \mathbf{x})$ be the predictions given by each random decision tree. $\A_K$ is $\left( \frac{4}{3}\right)$-training stable, if it computes the prediction as follows:
\begin{displaymath}
\p_{\A(T)}(y \mid \mathbf{x}) = e^{\frac{1}{K}\sum_{j = 1}^{K}\log(\p_j(y \mid \mathbf{x}))}.
\end{displaymath}
\end{proposition}
\begin{proof}
See appendix \ref{appendix:proof}.
\end{proof}

\paragraphbe{$k$-Nearest Neighbors.}
$k$-nearest neighbors ($k$-NN) classification\cite{altman1992introduction} is an instance-based learning algorithm. Instead of constructing a model from the training dataset, all examples in the training dataset are saved and all computations are deferred until classification. When responding to a query $(\mathbf{x},y)$, predictions are made by approximating locally from a few examples close to  the query. Unlike the aforementioned classifiers, $k$-NN is not training stable for any $\delta$.

For simplification, suppose $\A$ is a 1-nearest neighbor classification algorithm. Let $(\mathbf{x}^{(t)}, y^{(t)})$, $(\mathbf{x}_1,y_1)$ and  $(\mathbf{x}_2, y_2)$  be three examples in a training dataset. $(\mathbf{x}_1,y_1)$ is the nearest neighbor of $(\mathbf{x}^{(t)}, y^{(t)})$ when $t$ itself is not in the training dataset. Let $(\mathbf{x}', y')$ be a point whose nearest neighbor in the training dataset is $(\mathbf{x}^{(t)}, y^{(t)})$ and second nearest neighbor is $(\mathbf{x}_2, y_2)$. Suppose $y^{(t)} = y_1=y' \neq y_2$. When $t \in T$, the classifier $\A(T)$ will predict the class label as $y^{(t)}$ for both features $\mathbf{x}^{(t)}$ and $\mathbf{x}'$. When $t$ is removed from the training dataset, the classifier $\A(T \setminus \{t\})$ will still predict the class label as $y^{(t)}$ for feature $\mathbf{x}^{(t)}$ because of point $(\mathbf{x}_1,y_1)$. However, $\A(T \setminus \{t\})$ will predict the class label for $\mathbf{x}'$ as $y_2$ since it is closest to $(\mathbf{x}_2, y_2)$ when $(\mathbf{x}^{(t)}, y^{(t)})$ is removed. Consequently, when $t$ is removed from the training dataset, the prediction for $t$ remains unchanged, but the prediction for a neighboring point $(\mathbf{x}',y')$ is greatly influenced. If we calculate the probability given by the classifier, we have 
\begin{displaymath}
p_{\A(T)}(y' \mid \mathbf{x}') = 1\quad \textrm{and} \quad p_{\A(T \setminus \{t\})}(y' \mid \mathbf{x}') = 0,
\end{displaymath}
while
\begin{displaymath}
\frac{p_{\A(T)}(y^{(t)} \mid \mathbf{x}^{(t)})}{p_{\A(T \setminus \{t\})}(y^{(t)} \mid \mathbf{x}^{(t)})} = 1.
\end{displaymath}
As $k$ increases, the probability of the aforementioned case drops. However, there is always a possibility that, when an exampled $t$ is removed from the training dataset $T$, for some queries $(\mathbf{x},y)$, the prediction $p_{\A(T)}(y \mid \mathbf{x})$ drops from $1$ to $0$, while the prediction $p_{\A(T)}(y^{(t)} \mid \mathbf{x}^{(t)})$ remains unchanged. Therefore, $k$-nearest neighbors classification algorithm is not training stable.s

\subsection{An Upper Bound on DTP}
For non-training stable classifiers or classifiers with unknown training stability, the DTP metric cannot be directly calculated based on PDTP measurements. However, it is still possible to estimate an upper bound for DTP based on Lipschitz conditions. 

Given a classification algorithm $\A$, the set of possible classifiers learned by $\A$ can be abstracted as a class of functions $\{C_u, u \in \mathcal{U}\}$, where $u \in \mathcal{U}$ is a $d$-dimensional vector that specifies the trainable parameters in the classifier and $\mathcal{U} \subseteq \mathbb{R}^{d}$.  Without loss of generality, we assume that $C_u$ maps a feature vector $\mathbf{x}$ to a vector of predicted log probabilities of each class labels $y \in Y$. That is, $C_u(\mathbf{x}) = \left(\log p_{\A(T)}(y_1 \mid \mathbf{x}), \log p_{\A(T)}(y_2 \mid \mathbf{x}), \dots, \log p_{\A(T)}(y_k \mid \mathbf{x}) \right)$. 

We assume that for all $\mathbf{x} \in X^{m}$, $C_u(\mathbf{x})$ is $L$-Lipschitz bounded under infinity norm with respect to $u$. That is, $\left| C_u(\mathbf{x}) - C_{u'}(\mathbf{x}) \right|_{\infty} \leq L \left| u - u' \right|_{\infty}$. Based on Lipschitz condition, to calculate $\DTP(\A,T)$, it is enough to measure the change of model parameters when one training example is removed. 

Let $u_T$ be the model parameters learned on dataset $T$ and $u_{T \setminus \{t\}}$ be the parameters learned on $T \setminus \{t\}$. The following theorem gives an upper bound of $\DTP(\A,T)$:
\begin{theorem} 
\label{theorem:lipschitz}
If $\A(T)$ is an $L$-Lipschitz bounded classifier, then $\DTP_{\A,T}(t)$ is upper bounded by $L \cdot \max_{t \in T} \left| u_T - u_{T \setminus \{t\}} \right|_{\infty}$.
\end{theorem}

\begin{proof}
See Appendix~\ref{appendix:proof}.
\end{proof}

%% file: related_work.tex
\section{Related Work}
\label{sec:related}
Works on re-identification attacks~\cite{sweeney2013matching, loukides2010disclosure, narayanan2008robust} have demonstrated the privacy risks of releasing anonymized datasets. Releasing highly-reidentifiable datasets allows an attacker to infer sensitive attributes of individuals in these datasets. In these studies, the researchers collect background knowledge containing nonsensitive attributes from external sources, and use them to re-identify records with sensitive attributes in the anonymized datasets. In DTP, we make similar assumptions on the background knowledge, but instead of publishing the anonymized dataset, we assume that only the classification model learned on the dataset is published. We demonstrate that publishing a classification model learned from a highly-reidentifiable dataset can also bring high privacy risks. 

To protect against re-identification, syntactic privacy metrics like $k$-anonymity~\cite{sweeney2002k}, $l$-diversity~\cite{machanavajjhala2007diversity}, and $t$-closeness~\cite{li2007t} purports to measure the privacy risk of an anonymized dataset. These metrics reflect properties of a dataset, whereas DTP reflects the property of the combination of a classification algorithm and a dataset. 

Differential privacy (DP)~\cite{dwork2006differential}  guarantees privacy protection against an attacker with precise knowledge about the input dataset and \emph{all} the entities in the universe except for the target individual. Follow-up works on DP try to relax the background knowledge by building more realistic background knowledge models~\cite{machanavajjhala2009data, bassily2013coupled, li2012sampling, lee2012differential, tramer2015differential}. Most of these relaxations can be unified under the framework of membership privacy~\cite{li2013membership}, which shows that protecting private information is equivalent to preventing an attacker from knowing whether an individual is included in the input dataset. Specifically, DP is shown to be equivalent to membership privacy with mutually independent distributions. Other extensions on DP enhances protections on outliers in a data set while relaxes protections on the remaining examples~\cite{lui2015outlier, gehrke2012crowd}. Similar to DP, these privacy definitions cannot be directly calculated given a data set and an algorithm. DTP is different with them in terms that it can also be used as a privacy metric. 

DP has been widely applied to complete privacy preserving machine learning tasks with different machine learning algorithms such as regression~\cite{zhang2012functional}, SVM~\cite{rubinstein2009learning}, principal component analysis~\cite{chaudhuri2013near}, empirical risk minimization~\cite{chaudhuri2011differentially}, and deep learning~\cite{shokri2015privacy}. These privacy-preserving machine learning algorithms achieve DP by adding noises to either the objective function or the output parameters. A non-randomized machine learning algorithm cannot be $\epsilon$-DP or $\epsilon$-PMP for any $\epsilon$. DTP gives a possible way to measure the privacy risk of these non-randomized machine learning algorithms. 

Attacks on machine learning models have shown the possibility of inferring sensitive information about a model's training set with black-box access to the model. For example, model inversion attacks~\cite{fredrikson2014privacy} and naive attacks~\cite{cormode2011personal} infer sensitive attributes of a record based on the model's predictions on that record; membership inference attacks~\cite{shokri2016membership} infer whether a record is included in the training set of the model. Moreover, if the adversary colludes with the party that provides the training algorithm, he can embed sensitive information, including membership information, in the predictions of the models~\cite{song2017machine}. All these attacks can be summarized under the framework in~\cite{yeom2017unintended}. DTP extends this line of work by studying the \emph{unintentional} leakage of \emph{membership} information. We emprically estimate the risk of membership inference attacks under the assumption that the adversary only has black-box access to the model and does \emph{not} control the training algorithm. 

The connections between DP and overfitting is first pointed by the paper on the reusable hold out method~\cite{dwork2006differential}. The paper applies DP on the validation dataset to make it possible for the validation dataset to be reused, without the risk of overfitting. This demonstrates the possibility of using privacy preserving techniques to prevent overfitting. DTP makes this connection from the opposite direction: using cross validation techniques in overfitting prevention to measure the privacy risks. 

%% file: drop_dtp.tex
\section{Reducing DTP}

According to the DTP-$1$ hypothesis, it is unsafe to release a classifier if any record in its training set has DTP greater than 1. However, what should the data owner do if only a small number of records in the training dataset violate this hypothesis? It is unsafe to release the classifier since it contains records vulnerable to membership inference attacks using techniques like those in Section~\ref{sec:case_studies}. But it is natural to ask: {\em Can removing high-risk records from the training set mitigate the membership privacy risk?} The interesting answer is sometimes yes and sometimes no!

Let us consider a specific example of how removing high-risk records can influence DTP. The examples in Section~\ref{sec:case_studies} are not ideal because the classifiers trained on the adult dataset already have low privacy risks, while the classifiers trained on the purchased dataset are so risky that they are unlikely to be mitigated by simple mechanisms. We therefore consider a fresh example that fails to satisfy DTP-1, but not by much. To do this we train a naive Bayes classifier on the 2013 American Community Survey (ACS) dataset. This dataset has similar attributes as the adult dataset since the adult dataset was sampled and cleaned from the 1994 Census dataset. We restrict to four attributes: Age (AGEP), Marital Status (MAR), Race (RAC1P), and Gender (SEX). As we did with the adult dataset, we use the salary class ($>50$K or $\leq 50$K) as the class attribute. We use all $1.6$ million records as our training set. The DTP of the full dataset is $3.09$, indicating vulnerability to membership inference attacks.

To reduce the DTP in the dataset, we perform the following simple experiment: First, we measure DTP of each record in the training set and sort all the records in the decreasing order of their DTP measurements. Intuitively, the records are sorted in the decreasing order of their (initial) privacy risks. Next, we remove these high-risk records from the training set one at a time. After each record is removed, we re-calculate the DTP of all records remaining in the training set and estimate the resulting privacy risk as the highest DTP in this reduced training set. Figure~\ref{fig:remove}
\begin{figure}
    \includegraphics[width=0.9\columnwidth]{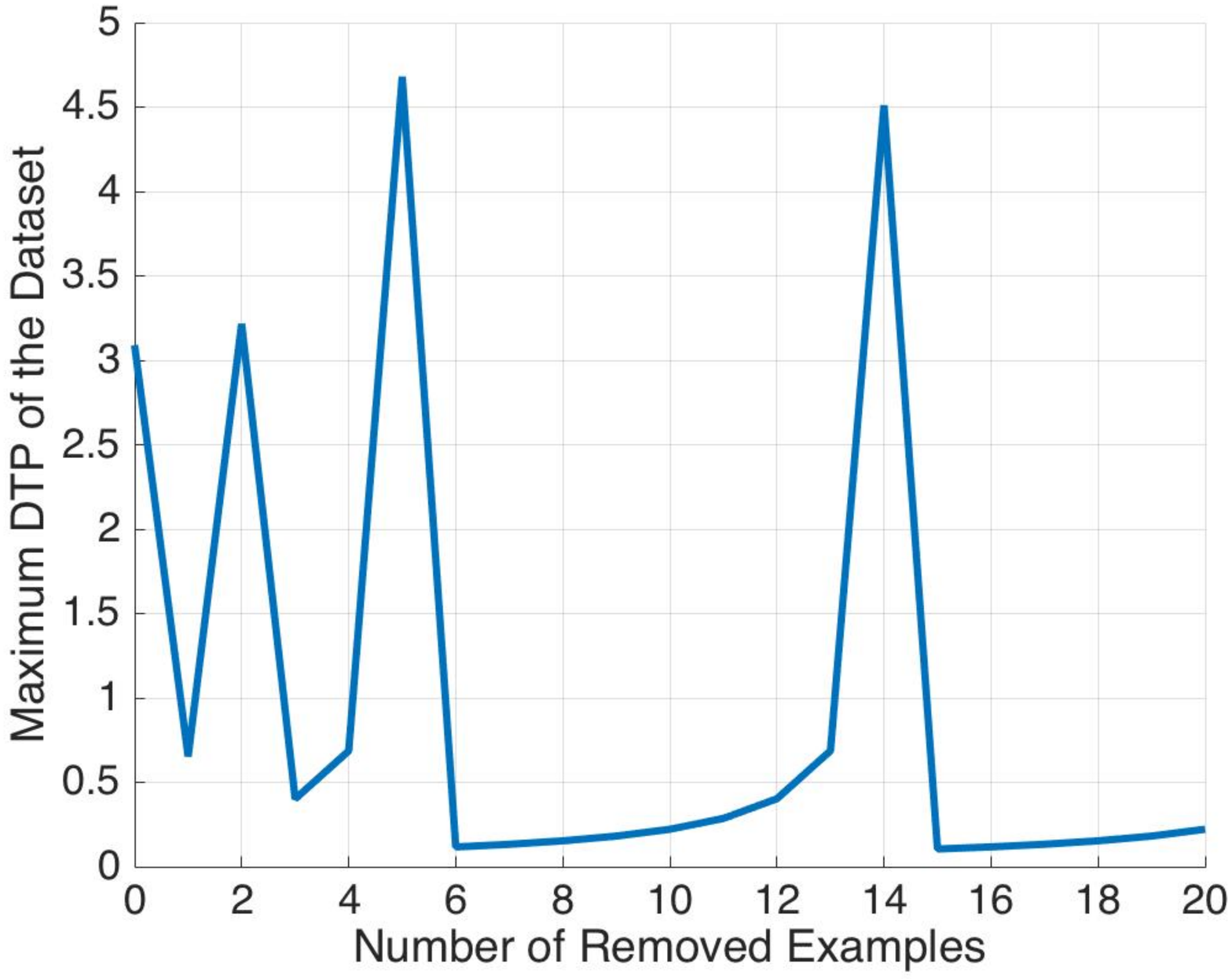}%
    \vspace{-3pt}
    \caption{\small Effects of Removing Hiigh-Risk Records. \vspace{-8pt}}%
    \label{fig:remove}%
\end{figure}%
shows the change of maximum DTP in the training set when these high-risk records are removed. Removing the record with the highest risk reduces the highest DTP in training set from $3.09$ to $0.65$, greatly reducing the classifier's vulnerability to membership attacks and achieving DTP-1. However, one must not get greedy and think that removing the next individual will reduce the risk even further. Doing this takes the DTP back to around $3$. Why? Because, unlike the first individual removed, this second record apparently is needed to decrease another record's influence on the target classifier. Indeed, removing further individuals appears to lead to collections of individuals that rely on each other to keep DTP down. Their successive removal creates the sawtooth pattern seen in Figure~\ref{fig:remove}. Based on this observation, we recommend removing high-risk examples as a way of reducing DTP and mitigating against membership attacks when only a few examples in the training set have high privacy risks. Better understanding of how to reduce DTP is a promising target for future research.

%% file: discussion.tex
\section{Discussion}

In this section, we discuss a few interesting points related to DTP.

\paragraphbe{Difference between DTP and DP.}
When feasible, using differential privacy during training is a good strategy to mitigate the risk of publishing the model.
However, there are cases when differential privacy cannot be used; either because there is no appropriate training mechanism or because the data owner cannot afford to add noise to their models (e.g., in the medical domain). 
Therefore, we need a strategy to estimate the privacy risk of the model when no privacy protections are added. 
Note that even when differential privacy is used, DTP can still be used to estimate the privacy risks before applying differential privacy. With the DTP measurements, the data owner can understand how much he benefits from using differential private mechanisms. This information helps balance the trade-off between utility and privacy.

Unlike DP, DTP is a privacy metric instead of a privacy protection mechanism. When a machine learning model does not satisfy differential privacy for any $\epsilon$, little is known about its privacy risk. However, the metric $\DTP_{\A,T}(t)$ outputs a value of $\epsilon$ for \emph{any} target record $t$ and any \emph{any} classifier $\A(T)$. 

\paragraphbe{Difference between DTP and Membership Attacks.}
In Section~\ref{sec:case_studies}, we show that DTP measurements correlate with the accuracy of different membership attacks. However, the measurement of DTP \emph{cannot} be replaced by running a series of membership attacks. First, it is computationally inefficient to simulate all possible membership attacks. Moreover, no matter how much computational power a data owner has, there may always exist an adversary with superior computational capability. Second, we cannot rule out the possibility that the adversary knows a stronger a membership attack than the data owner. As demonstrated in section~\ref{sec:case_studies}, a record immune to the distance-based membership attack can be vulnerable to another attack---even a weaker one, overall. Therefore, using a general privacy metric like DTP to estimate membership privacy is preferable.

\paragraphbe{Privacy Risks of Non-Training Stable Classifiers.}
In Section~\ref{sec:measure_dtp}, we prove that naive Bayes, random decision trees, and linear statistical queries satisfy training stability but $k$-NN provably does not. However, we do not know whether classifiers such as neural networks and SVMs are training stable. We leave the task of investigating this question for future work.

Remark that although measuring DTP is computationally infeasible for non-training stable models, this does not mean that DTP metrics are useless for these models. Indeed, as shown in Section~\ref{sec:case_studies} PDTP measurements have high correlations with direct membership attacks. The drawback for non-training stable algorithms is their potential vulnerability to indirect membership attacks. As future work, we plan to study indirect membership attacks and ways to mitigate them.

Although DTP doesn't provide a theoretical privacy guarantee like DP, we find that it is highly correlated with the performance of state-of-the-art membership inference techniques. Unlike DP, which bounds the change in the probability of observing an output when a record is removed, DTP bounds the magnitude of the difference caused by removal of a record. In practice, if the magnitude of this difference is small, it is indistinguishable from the difference caused by other uncertain factors from the adversary's perspective. When attacking machine learning models, these are at least two sources of uncertainty.

First, in models like neural networks, some parameters such as weights are initialized randomly. Different initialization states may cause the model to be converged to different local optimals, so models trained on the same dataset can give slightly different predictions on the same record. If a record's DTP is small enough to be indistinguishable from the difference caused by random initialization, the record has little privacy risk. Figure~\ref{fig:init_var} shows the variation in model prediction caused by random initialization. In the experiment, we train $100$ neural network models on the same training set with $10000$ records uniformly sampled from the UCI Adult dataset. We calculate each model's prediction on two individuals and plot the histogram of the predicted probability that the individual has annual salary greater than 50K. 

\begin{figure}
    \centering
    \begin{subfigure}[b]{0.24\textwidth}
        \includegraphics[width=\textwidth]{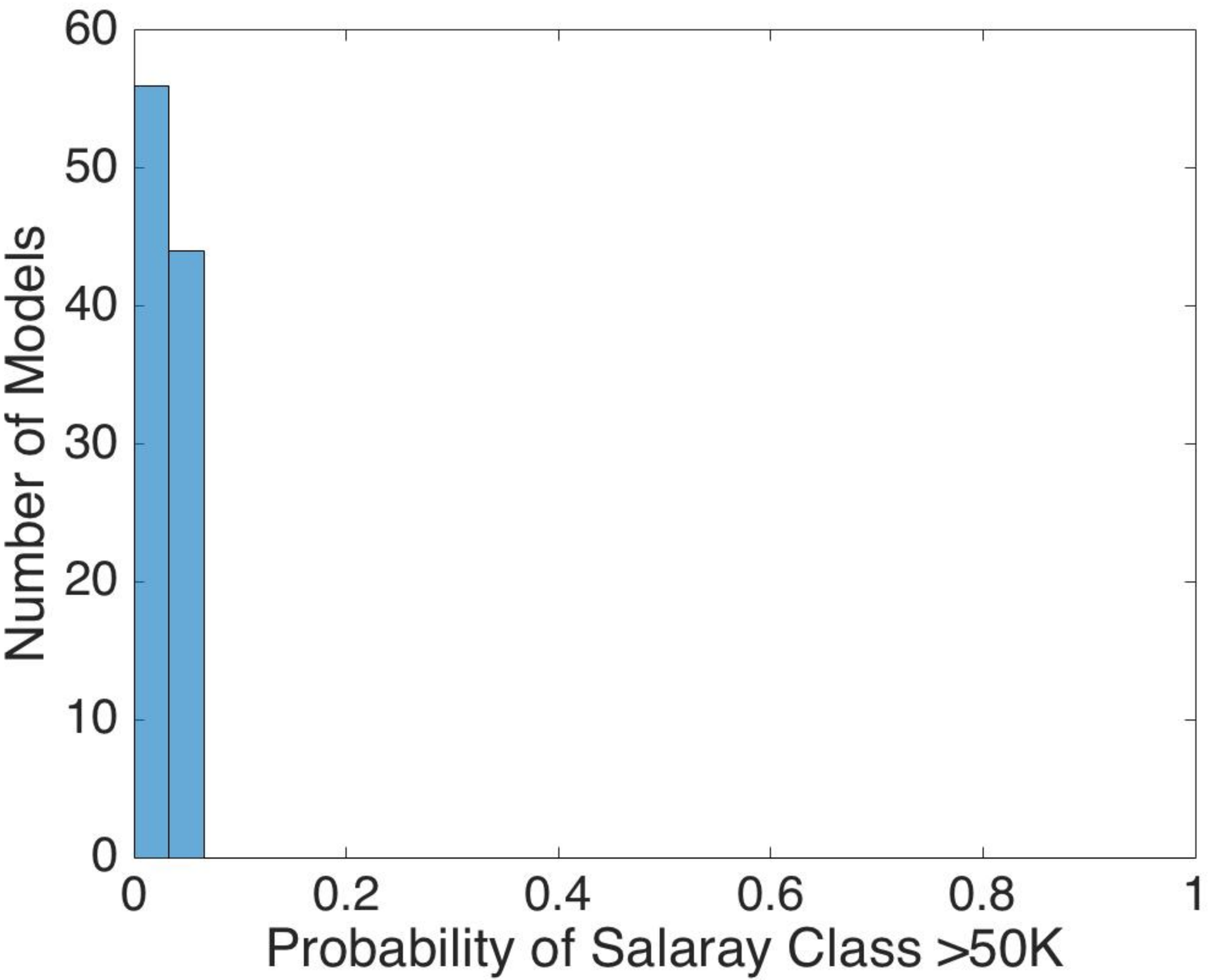}
        \caption{Prediction on an individual in low salary class}
    \end{subfigure}
    ~ 
    \begin{subfigure}[b]{0.24\textwidth}
        \includegraphics[width=\textwidth]{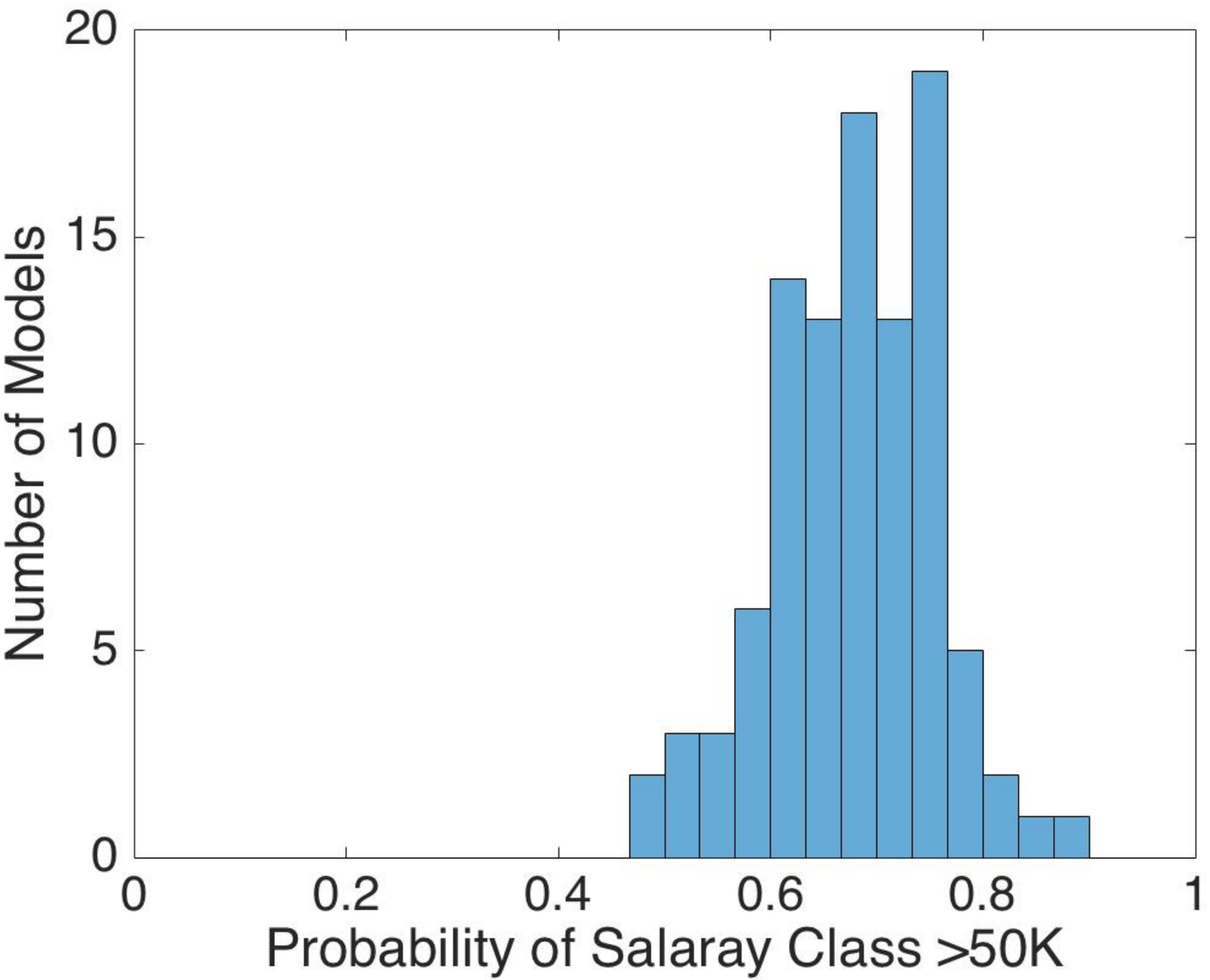}
        \caption{Prediction on an individual in high salary class}
    \end{subfigure}
    ~ 
    \caption{Prediction variation caused by random initialization}\label{fig:init_var}
\end{figure}

Second, besides the target record, the adversary is also uncertain about what other records are included in the training dataset. The occurrence of unexpected training records can introduce small variations in the model's predcitions. If a record's DTP is small enough to be indistinguishable from the variation caused by random sampling, the target has little privacy risk. Figure~\ref{fig:sampling_var} shows the variation in model prediction caused by random sampling. In the experiment, we train $100$ classifiers on $100$ different training datasets uniformly sampled from the same population. We calculate each model's prediction on two individuals and plot the histogram of the predicted probability that the individual has annual salary greater than 50K. 

\begin{figure}
    \centering
    \begin{subfigure}[b]{0.24\textwidth}
        \includegraphics[width=\textwidth]{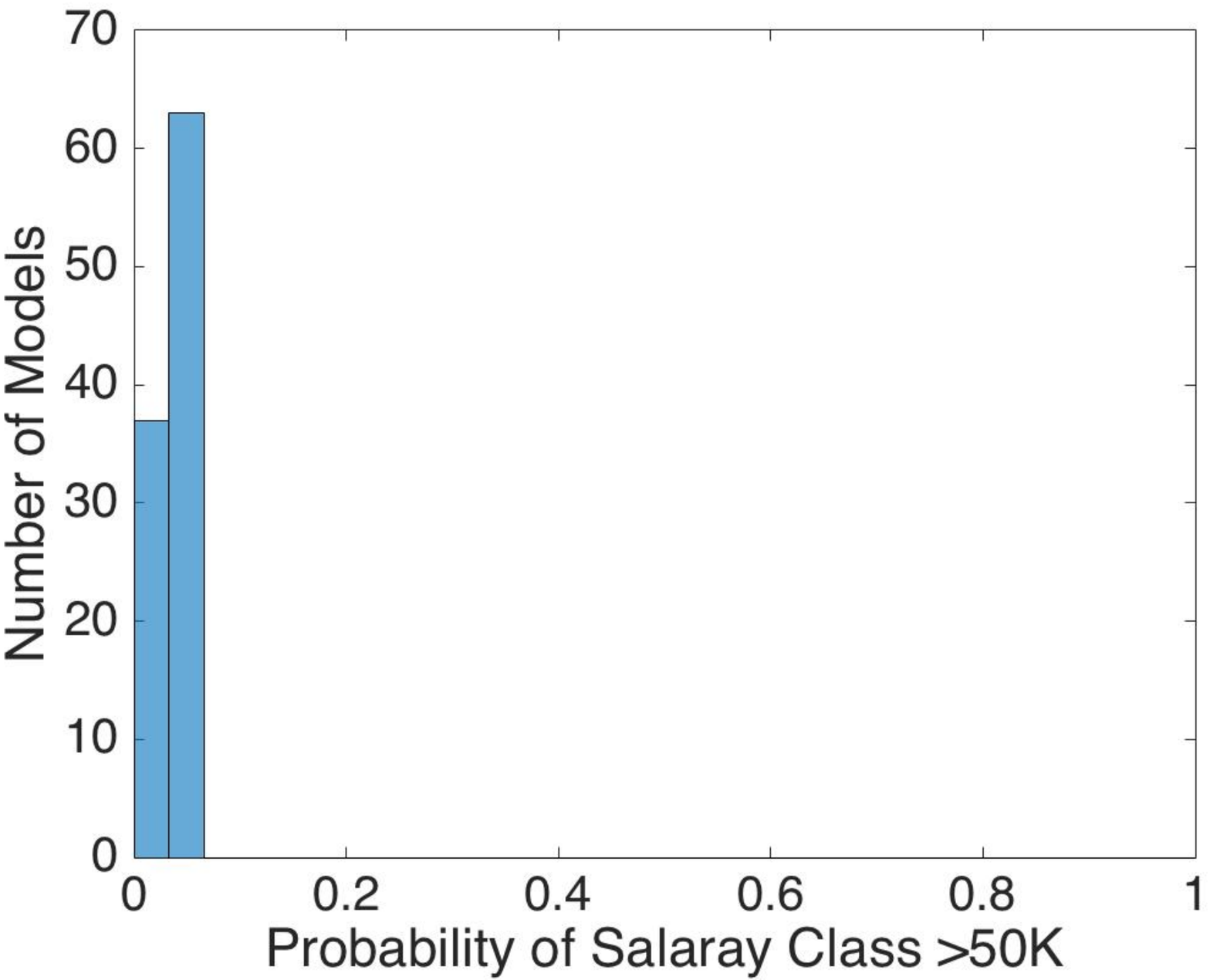}
        \caption{Prediction on an individual in low salary class}
    \end{subfigure}
    ~ 
    \begin{subfigure}[b]{0.24\textwidth}
        \includegraphics[width=\textwidth]{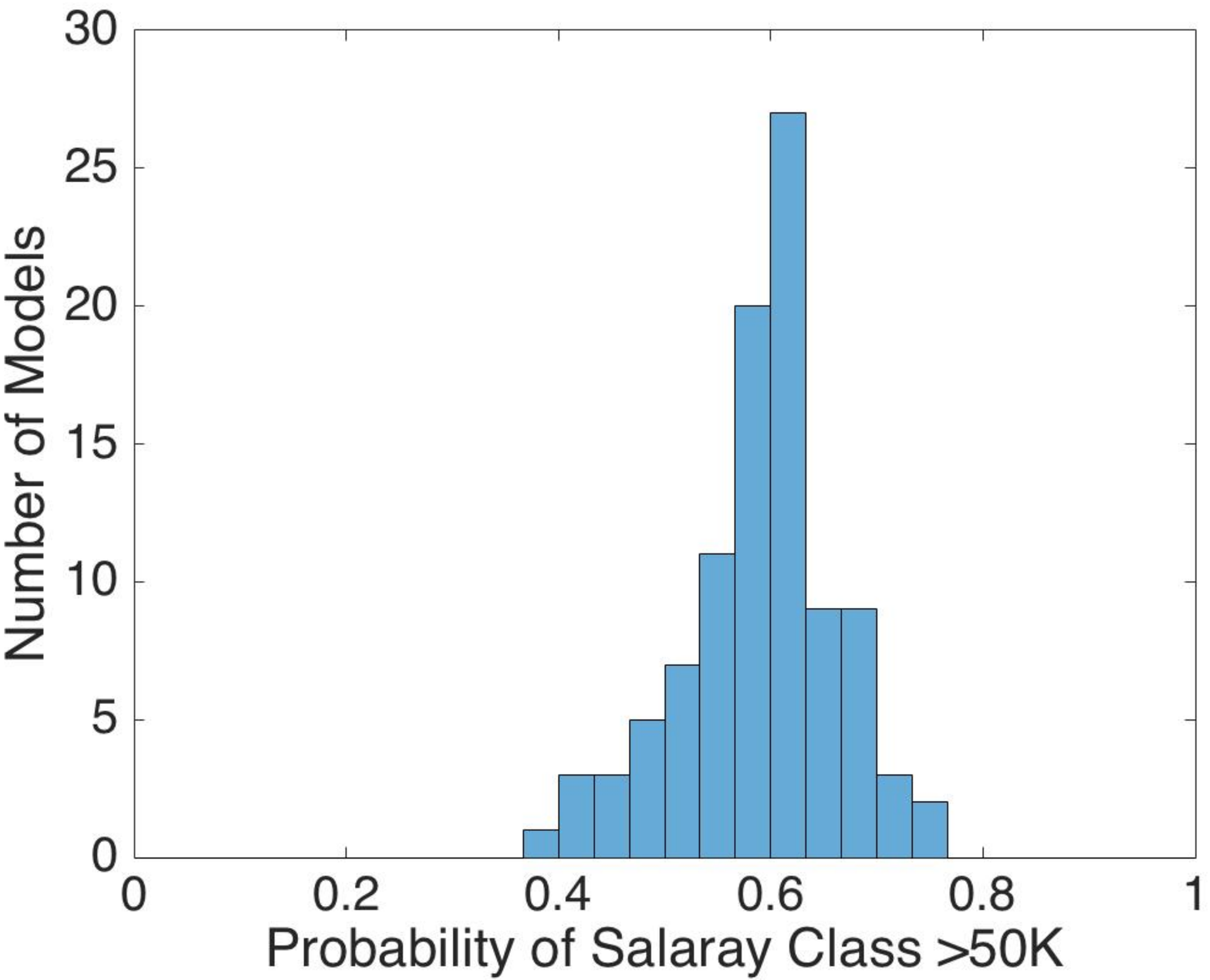}
        \caption{Prediction on an individual in high salary class}
    \end{subfigure}
    ~ 
    \caption{Prediction variation caused by random sampling}\label{fig:sampling_var}
\end{figure}

%% file: open_questions.tex
\section{Open Questions}

Experimental results suggest that DTP is a good predictor of the performance of state-of-art membership inference attacks. However, it remains an open question if records with low DTP are always safe from membership inference attacks. In this section, we discuss two potential privacy risks for low-DTP records.

\paragraphbe{DTP-$1$ Hypothesis.}
In this paper, we use the DTP-$1$ hypothesis as guidance to identify records and classifiers with high privacy risk. By experimenting with state-of-art membership inference attacks on machine learning models, we find that when a training record only has a very small influence on the prediction of a classifier, this small influence is likely to be indistinguishable from the variation in prediction due to random sampling of the training records or random initialization of the weight vectors before training. We use DTP-$1$ hypothesis as a rule-of-thumb for determining whether the influence of a training record is smaller than the influence of other factors unknown to the adversary, such as randomization in the training algorithm and existence of unexpected records in the training set. However, in practice, even when DTP is smaller than $1$, the influence of these uncertain factors can be smaller than the influence of the training record. Therefore, satisfying the DTP-1 hypothesis cannot guarantee that records with DTP smaller than 1 have no privacy risks. With a better understanding on the adversary's background knowledge and the influence of randomness in machine learning algorithms, it may be possible to determine a finer threshold for a safe DTP.

\paragraphbe{Risk of Indirect Attacks and Multiple Queries.}
PDTP measures the privacy risk of directly querying the target record. Based on experimental validations, we find that training records with low PDTP are less likely to be vulnerable to direct attacks. However, models that are not training stable have a potential of leaking the record's membership information through other queries, and an adversary may use this information to perform an indirect attack. Although we do not know of any practical indirect attacks, it remains an open question to analyze the training stability of some machine learning models and to design indirect attacks for models that are not training stable.

Another challenge is to analyze the risk of allowing an adversary to get predictions of multiple queries from the same machine leanrning model. DTP measures the privacy risk for a single query. However, if an adversary is allowed to query the target model multiple times, he may accumulate more information about the target record $t$. We leave for future work the study of how this accumulation of information can be used to design stronger membership inference attacks. Specifically, there are two open questions: (1) \emph{How do we select multiple queries whose results indicate the membership of a target record?} (2) \emph{How do we estimate an upper bound on the accuracy of membership inference when an adversary can submit unlimited number of queries to the model?}

%% file: conclusion.tex
\section{Conclusions}

In this work, we propose differential training privacy (DTP) as an empirical metric to estimate the privacy risk of publishing a classifier. DTP estimates the privacy risk of a training record by measuring its influence on the predictions of machine learning models. A large DTP indicates that the record's influence is strong enough to indicate its presence in the training dataset. We measure DTP of popular machine learning models including neural networks, Naive Bayes, and logistic regressions. We compare these measurements with the accuracy of different types of membership inference attacks, including the most effective one in prior works. Experimental results demontrate that DTP is both efficient and effective in estimating privacy risks. Specifically, our attacks have at most $66.5\%$ accuracy (baseline: $50\%$) on classifiers with DTP-values under $0.5$ and almost always over $90\%$ accuracy on classifiers with DTP larger than $4$. Based on these results, we propose DTP-$1$ hypothesis as a rule-of-thumb criterion for publishing a classifier: \emph{if a classifier has a DTP value above $1$, it should not be published. }

Although DTP has a high correlation with the accuracy of a membership attack, it provides no guarantee about a record's privacy protection. Specifically, we propose two potential privacy leakages for records with low DTP. First, a low-DTP record is vulnerable to membership inferences when the model's predictions on it are unlikely to be influenced by other records or random initializations. Second, the membership of a low-DTP record might be leaked by indirect queries or the combination of multiple queries. This observation can serve as a new direction for designing stronger membership inference attacks and defenses. 

%% file: appendix.tex
\pagebreak

\section{Proofs}

\label{appendix:proof}

\begin{proof}[Proof of Theorem~\ref{theorem:stability}]
Since $(\A,T)$ is $\epsilon$-PDTP, for all $t \in T$, we have
\begin{displaymath}
	p_{\A(T)}(y^{(t)} \mid \mathbf{x}^{(t)}) \le e^{\epsilon} p_{\A(T \setminus \{t\})}(y^{(t)} \mid \mathbf{x}^{(t)})
\end{displaymath}
and
\begin{displaymath}
	p_{\A(T)}(y^{(t)} \mid \mathbf{x}^{(t)}) \ge e^{-\epsilon} p_{\A(T \setminus \{t\})}(y^{(t)} \mid \mathbf{x}^{(t)}).
\end{displaymath}
Since $\A$ is $\delta$-training stable on $T$, for $\gamma = \max(\delta, e^{\epsilon})$, we have
\begin{displaymath}
p_{\A(T)}(y \mid \mathbf{x}) \le \gamma p_{\A(T \setminus \{t\})}(y \mid \mathbf{x})
\end{displaymath}
and
\begin{displaymath}
p_{\A(T)}(y \mid \mathbf{x}) \ge \gamma^{-1} p_{\A(T \setminus \{t\})}(y \mid \mathbf{x}).
\end{displaymath}
Let $\epsilon ' = \ln(\gamma) = \max(\ln(\delta), \epsilon)$, then
\begin{displaymath}
p_{\A(T)}(y \mid \mathbf{x}) \le e^{\epsilon'} p_{\A(T \setminus \{t\})}(y \mid \mathbf{x})
\end{displaymath}
and
\begin{displaymath}
p_{\A(T)}(y \mid \mathbf{x}) \ge e^{-\epsilon'} p_{\A(T \setminus \{t\})}(y \mid \mathbf{x}).
\end{displaymath}
Therefore, $(\A,T)$ is $\epsilon '$-DTP.
\end{proof}

\begin{proof}[Proof of Proposition~\ref{prop:bi}]
Let $T' = T \setminus \{t\}$.  Let $\A(T)$  and $\A(T')$ be two Naive Bayes classification models separately learned on dataset $T$ and $T'$, and $N$ be the number of records in $T$. We use $n_{\mathbf{x}, y}$ to represent the number of records with the feature vector $\mathbf{x}$ and the class label $y$, and use $n_\mathbf{x}$ to represent the number of records in $T$ with feature vector $\mathbf{x}$.
For all $\mathbf{x} \in X^{m}$, for all $y \in Y$, so that $n_{\mathbf{x},y} > 1$ and $n_{\mathbf{x}} > 1$, we have
\begin{displaymath}
p_{\A(T)}(y \mid \mathbf{x}) = \frac{n_{\mathbf{x},y}}{n_{\mathbf{x}}}.
\end{displaymath}
When $\mathbf{x} \neq \mathbf{x}^{(t)}$,
\begin{equation}
\label{eq:x_neq}
p_{\A(T \setminus \{t\})}(y \mid \mathbf{x}) = p_{\A(T)}(y \mid \mathbf{x}).
\end{equation}
When $\mathbf{x} = \mathbf{x}^{(t)}$ and $y \neq y^{(t)}$,
\begin{displaymath}
p_{\A(T \setminus \{t\})}(y \mid \mathbf{x}^{(t)}) = \frac{n_{\mathbf{x}^{(t)},y}}{n_{\mathbf{x}^{(t)}} - 1}.
\end{displaymath}
Therefore,
\begin{equation}
\label{eq:y_neq}
\frac{p_{\A(T)}(y \mid \mathbf{x}^{(t)})}{p_{\A(T \setminus \{t\})}(y \mid \mathbf{x}^{(t)})} = \frac{n_{\mathbf{x}^{(t)}} - 1}{n_{\mathbf{x}^{(t)}}} < 1.
\end{equation}
When $\mathbf{x} = \mathbf{x}^{(t)}$ and $y = y^{(t)}$,
\begin{displaymath}
p_{\A(T \setminus \{t\})}(y^{(t)} \mid \mathbf{x}^{(t)}) = \frac{n_{\mathbf{x}^{(t)},y^{(t)}} - 1}{n_{\mathbf{x}^{(t)}} - 1}.
\end{displaymath}
Therefore,
\begin{displaymath}
\begin{aligned}
\frac{p_{\A(T)}(y^{(t)} \mid \mathbf{x}^{(t)})}{p_{\A(T \setminus \{t\})}(y^{(t)} \mid \mathbf{x}^{(t)})} 
&= \frac{n_{\mathbf{x}^{(t)},y^{(t)}}}{n_{\mathbf{x}^{(t)}}} \left(\frac{n_{\mathbf{x}^{(t)},y^{(t)}}-1}{n_{\mathbf{x}^{(t)}}-1}\right)^{-1} \\
&= \frac{n_{\mathbf{x}^{(t)},y^{(t)}}n_{\mathbf{x}^{(t)}} - n_{\mathbf{x}^{(t)},y^{(t)}}}
{n_{\mathbf{x}^{(t)},y^{(t)}}n_{\mathbf{x}^{(t)}} - n_{\mathbf{x}^{(t)}}}.
\end{aligned}
\end{displaymath}
Since
\begin{displaymath}
n_{\mathbf{x}^{(t)},y^{(t)}} \leq n_{\mathbf{x}^{(t)}},
\end{displaymath}
we have
\begin{equation}
\label{eq:xy_eq}
\frac{p_{\A(T)}(y^{(t)} \mid \mathbf{x}^{(t)})}{p_{\A(T \setminus \{t\})}(y^{(t)} \mid \mathbf{x}^{(t)})} \geq 1.
\end{equation}

From equations~\ref{eq:x_neq}, \ref{eq:y_neq}, and \ref{eq:xy_eq}, we have for all $\mathbf{x} \in X^{m}$, for all $y \in Y$,
\begin{displaymath}
\frac{p_{\A(T)}(y \mid \mathbf{x})}{p_{\A(T \setminus \{t\})}(y \mid \mathbf{x})} \le
\frac{p_{\A(T)}(y^{(t)} \mid \mathbf{x}^{(t)})}{p_{\A(T \setminus \{t\})}(y^{(t)} \mid \mathbf{x}^{(t)})}. \end{displaymath}
Therefore, it suffices to prove
\begin{displaymath}
\frac{p_{\A(T)}(y \mid \mathbf{x})}{p_{\A(T \setminus \{t\})}(y \mid \mathbf{x})} \ge
\left(\max( \frac{4}{3}, \frac{p_{\A(T)}(y^{(t)} \mid \mathbf{x}^{(t)})}{p_{\A(T \setminus \{t\})}(y^{(t)} \mid \mathbf{x}^{(t)})})\right)^{-1},
\end{displaymath}
which is equivalent to
\begin{displaymath}
\frac{n_{\mathbf{x}^{(t)}}}{n_{\mathbf{x}^{(t)}} -1 } \le \max(\frac{4}{3}, \frac{n_{\mathbf{x}^{(t)},y^{(t)}}}{n_{\mathbf{x}^{(t)},y^{(t)}} - 1}\frac{n_{\mathbf{x}^{(t)}} - 1}{n_{\mathbf{x}^{(t)}}}).
\end{displaymath}
Let $a = n_{\mathbf{x}^{(t)},y^{(t)}}$ and $b = n_{\mathbf{x}^{(t)}}$. Then $a < b$.
By solving
\begin{displaymath}
\frac{b}{b-1} > \frac{b-1}{b}\frac{a}{a-1},
\end{displaymath}
we get
\begin{displaymath}
b > a + \frac{\sqrt{4a(a-1)}}{2}.
\end{displaymath}
Let $b^* = \lceil a + \frac{\sqrt{4a(a-1)}}{2} \rceil$, then $b^*_{\min} = 4$. Since $\frac{b}{b-1}$ is a decreasing function of $b$, when $\frac{b}{b-1} > \frac{b-1}{b}\frac{a}{a-1}$, we have
\begin{displaymath}
\frac{b}{b-1} \le \frac{b^*_{\min}}{b^*_{\min}-1} = \frac{4}{3}.
\end{displaymath}
Therefore,
\begin{displaymath}
\frac{n_{\mathbf{x}^{(t)}}}{n_{\mathbf{x}^{(t)}} -1 } \le \max(\frac{4}{3}, \frac{n_{\mathbf{x}^{(t)},y^{(t)}}}{n_{\mathbf{x}^{(t)},y^{(t)}} - 1}\frac{n_{\mathbf{x}^{(t)}} - 1}{n_{\mathbf{x}^{(t)}}}).
\end{displaymath}
Bayes inference algorithm is $\frac{4}{3}$-training stable.
\end{proof}

\begin{proof}[Proof of Proposition~\ref{prop:nb}]
Let $T' = T \setminus \{t\}$. Let $\A(T)$  and $\A(T')$ be two Bayesian classification models separately learned on dataset $T$ and $T'$, and $N$ be the number of records in $T$.  For simplification, we use $\p(x_i \mid y)$ and $\p'(x_i \mid y)$ to represent the conditional probability that the $i$-th feature equals to $x_i$ given the class label equals $y$ in $T$ and $T'$, and use $\p(y)$ and $\p'(y)$ to represent probability that the class label equals $y$ in $T$ and $T'$. We use $n_{x_i, y}$ to represent the number of records with the $i$-th feature $x_i$ and the class label $y$, and use $n_y$ to represent the number of records in $T$ with label $y$.

According to the conditional independence assumption of Naive Bayes, we have
\begin{equation}
\label{eq:nb}
\frac{\p(y \mid \mathbf{x})}{\p'(y \mid \mathbf{x})} = \frac{\prod_{i=1}^{m} \p(x_i \mid y) \p(y)}{\prod_{i=1}^{m} \p'(x_i \mid y) \p'(y)} = \prod_{i=1}^{m} \frac{\p(x_i \mid y)}{\p'(x_i \mid y)} \frac{\p(y)}{\p'(y)}.
\end{equation}

First, we prove that
\begin{equation*}
\frac{\p(y^{(t)} \mid \mathbf{x}^{(t)})}{\p'(y^{(t)} \mid \mathbf{x}^{(t)})} \geq 1.
\end{equation*}
Since $\p'(y^{(t)} \mid \mathbf{x}^{(t)}) > 0$, based on Equation~\ref{eq:nb}, we have $\p'(x^{(t)}_i \mid y^{(t)}) > 0$, and $\p'(y) > 0$.
\begin{equation*}
\begin{aligned}
\frac{\p(x^{(t)}_i \mid y^{(t)})}{\p'(x^{(t)}_i \mid y^{(t)})}
= & \frac{n_{x^{(t)}_i, y^{(t)}}\left(n_{y^{(t)}}-1\right)}{n_{y^{(t)}}\left(n_{x^{(t)}_i, y^{(t)}}-1\right)} \\
= & \frac{n_{x^{(t)}_i, y^{(t)}}n_{y^{(t)}} - n_{x^{(t)}_i, y^{(t)}}}{n_{x^{(t)}_i, y^{(t)}}n_{y^{(t)}} - n_{y^{(t)}}}.
\end{aligned}
\end{equation*}
Since
\begin{equation*}
n_{x^{(t)}_i, y^{(t)}} \le n_{y^{(t)}},
\end{equation*}
we have
\begin{equation}
\label{eq:pxy_mono}
\frac{\p(x_i^{(t)} \mid y^{(t)})}{\p'(x_i^{(t)} \mid y^{(t)})} \geq 1.
\end{equation}
Similarly
\begin{equation*}
\begin{aligned}
\frac{\p(y)}{\p'(y)} = &\frac{\p(y^{(t)})}{\p'(y^{(t)})}
= & \frac{n_{y^{(t)}}N - n_{y^{(t)}}}{n_{y^{(t)}}N - N} \geq 1.
\end{aligned}
\end{equation*}
Therefore,
\begin{equation}
\label{eq:ytxt}
\frac{\p(y^{(t)} \mid \mathbf{x}^{(t)})}{\p'(y^{(t)} \mid \mathbf{x}^{(t)})}
= \prod_{i=1}^{m}\frac{\p(x_i^{(t)} \mid y )}{\p'(x_i^{(t)} \mid y)} \frac{\p(y^{(t)})}{\p'(y^{(t)})} \geq 1.
\end{equation}

Next, we prove that Naive Bayes classification algorithm is $\delta$-training stable on $T$, for
\begin{equation*}
\delta = \left(\frac{n_{y_{\min}}}{n_{y_{\min}} - 1}\right)^{m-1} \frac{n}{n - 1}.
\end{equation*}
We start with the case when $y \neq y^{(t)}$.
Since, $n_{x_i, y}$ and $n_y$ does not change after the removal of $d$, we have
\begin{equation*}
\frac{\p(x_i \mid y)}{\p'(x_i \mid y)} = 1.
\end{equation*}
And,
\begin{equation*}
\begin{aligned}
&\frac{\p(y)}{\p'(y)}
= & \frac{n_y}{n} \left(\frac{n_y}{n-1}\right)^{-1}
= \frac{n-1}{n}.
\end{aligned}
\end{equation*}
Therefore, when  $y \neq y^{(t)}$
\begin{equation*}
\begin{aligned}
&\frac{\p(y \mid \mathbf{x})}{\p'(y\mid \mathbf{x})}
= &\prod_{i=1}^{m}\frac{\p(x_i\mid y)}{\p'(x_i \mid y)} \frac{\p(y\mid \mathbf{x})}{\p'(y\mid\mathbf{x})}
= &\frac{n-1}{n}< 1.
\end{aligned}
\end{equation*}
According to equation~\ref{eq:ytxt}, we have
\begin{equation*}
\begin{aligned}
\frac{\p(y \mid \mathbf{x})}{\p'(y\mid \mathbf{x})}
< \frac{\p(y^{(t)} \mid \mathbf{x}^{(t)})}{\p'(y^{(t)} \mid \mathbf{x}^{(t)})}.
\end{aligned}
\end{equation*}
Since
\begin{equation*}
\frac{n-1}{n} > \left(\frac{n_{y_{\min}} - 1}{n_{y_{\min}}}\right)^{m-1} \frac{n-1}{n} \geq \delta^{-1},
\end{equation*}
we have
\begin{equation*}
\delta^{-1} < \frac{\p(y \mid \mathbf{x})}{\p'(y\mid \mathbf{x})} < \frac{\p(y^{(t)} \mid \mathbf{x}^{(t)})}{\p'(y^{(t)} \mid \mathbf{x}^{(t)})}.
\end{equation*}
We then consider the case when $y = y^{(t)}$.\\
When $x_i = x_i^{(t)}$,
\begin{equation}
\label{eq:pxy_1}
 \frac{\p(x_i \mid y^{(t)})}{\p'(x_i \mid y^{(t)})} =  \frac{\p(x_i^{(t)} \mid y^{(t)})}{\p'(x_i^{(t)} \mid y^{(t)})} \geq 1.
\end{equation}
When $x_i \neq x_i^{(t)}$,
\begin{equation}
\label{eq:pxy_2}
\begin{aligned}
\frac{\p(x_i \mid y^{(t)})}{\p'(x_i \mid y^{(t)})} = \frac{n_{y^{(t)}}-1}{n_{y^{(t)}}} < 1.
\end{aligned}
\end{equation}
According to Equation~\ref{eq:pxy_mono}, we have
\begin{equation*}
\frac{\p(x_i \mid y^{(t)})}{\p'(x_i \mid y^{(t)})}  = \frac{n_{y^{(t)}}-1}{n_{y^{(t)}}} < \frac{\p(x_i^{(t)} \mid y^{(t)})}{\p'(x_i^{(t)} \mid y^{(t)})}.
\end{equation*}
Therefore,
\begin{equation*}
\begin{aligned}
\frac{\p(y^{(t)} \mid \mathbf{x})}{\p'(y^{(t)}\mid \mathbf{x})}
&= \prod_{i=1}^{m} \frac{\p(x_i \mid y^{(t)})}{\p'(x_i \mid y^{(t)})} \frac{\p(y^{(t)})}{\p'(y^{(t)})} \\
&\leq \prod_{i=1}^{m} \frac{\p(x_i^{(t)} \mid y^{(t)})}{\p'(x_i^{(t)}  \mid y^{(t)})} \frac{\p(y^{(t)})}{\p'(y^{(t)})} \\
&=  \frac{\p(y^{(t)} \mid \mathbf{x}^{(t)})}{\p'(y^{(t)} \mid \mathbf{x}^{(t)})}.
\end{aligned}
\end{equation*}
From Equation~\ref{eq:pxy_1} and \ref{eq:pxy_2}, we have
\begin{equation*}
\prod_{i=1}^{m} \frac{\p(x_i \mid y^{(t)})}{\p'(x_i \mid y^{(t)})} \geq \left(\frac{n_{y^{(t)}}-1}{n_{y^{(t)}}}\right)^{m}
\end{equation*}
Since
\begin{equation*}
\frac{\p(y^{(t)})}{\p'(y^{(t)})} = \frac{n_{y^{(t)}} (n-1)}{n(n_{y^{(t)}} - 1)} = \frac{n_{y^{(t)}}}{n_{y^{(t)}} - 1} \frac{n-1}{n},
\end{equation*}
we have
\begin{equation*}
\begin{aligned}
\frac{\p(y^{(t)} \mid \mathbf{x})}{\p'(y^{(t)}\mid \mathbf{x})}
&= \prod_{i=1}^{m} \frac{\p(x_i \mid y)}{\p'(x_i \mid y)} \frac{\p(y^{(t)})}{\p'(y^{(t)})} \\
&\geq \left(\frac{n_{y^{(t)}}-1)}{n_{y^{(t)}}}\right)^{m-1} \frac{n-1}{n} \\
&\geq \delta^{-1}.
\end{aligned}
\end{equation*}
Hence,
\begin{equation*}
\delta^{-1} \leq \frac{\p(y^{(t)} \mid \mathbf{x})}{\p'(y^{(t)}\mid \mathbf{x})} \leq \frac{\p(y^{(t)} \mid \mathbf{x}^{(t)})}{\p'(y^{(t)} \mid \mathbf{x}^{(t)})}.
\end{equation*}
Let $\gamma = \max\left(\delta,\frac{\p(y^{(t)} \mid \mathbf{x}^{(t)})}{\p'(y^{(t)} \mid \mathbf{x}^{(t)})}\right)$, then
\begin{equation*}
\gamma^{-1} \leq \frac{\p(y \mid \mathbf{x})}{\p'(y\mid \mathbf{x})} \leq \gamma.
\end{equation*}
\end{proof}

\begin{proof}[Proof of Proposition~\ref{prop:dt}]
Let $\p'_1(y \mid \mathbf{x}), \p'_2(y \mid \mathbf{x}), \dots, \p'_k(y \mid \mathbf{x})$ be the predictions given by each random decision tree learned on $T \setminus \{t\}$. Then for $1 \le k \le K$, we have
\begin{displaymath}
\frac{\p_k(y \mid \mathbf{x})}{\p'_k(y \mid \mathbf{x})} \le \max\left(\frac{4}{3},\frac{\p_k(y^{(t)} \mid \mathbf{x}^{(t)})}{\p'_K(y^{(t)} \mid \mathbf{x}^{(t)})}\right)
\end{displaymath}
Therefore,
\begin{displaymath}
\begin{aligned}
\frac{\p_{\A(T)}(y \mid \mathbf{x})}{\p_{\A(T \setminus \{t\})}(y \mid \mathbf{x})}
= &\left(\prod_{k = 1}^{K}\frac{\p_k(y \mid \mathbf{x})}{\p'_k(y \mid \mathbf{x})} \right)^{\frac{1}{K}} \\
\le & \max \left(\frac{4}{3}, \frac{\p_{\A(T)}(y^{(t)} \mid \mathbf{x}^{(t)})}{\p_{\A(T \setminus \{t\})}(y^{(t)} \mid \mathbf{x}^{(t)})} \right)
\end{aligned}
\end{displaymath}
\end{proof}

\begin{proof}[Proof of Theorem~\ref{theorem:lipschitz}]
According to Lipschitz condition, for all $\mathbf{x} \in X^{m}$, for all $t \in T$, we have
\begin{displaymath}
\left| C_{u_T}(\mathbf{x}) - C_{u_T^{-t}}(\mathbf{x}) \right|_{\infty} \leq L \left| u_T - u_T^{-t} \right|_{\infty} \leq \epsilon L.
\end{displaymath}
That is, for all $y \in Y$,
\begin{displaymath}
\left| \log p_{\A(T)}(y \mid \mathbf{x}) - \log p_{\A(T \setminus \{t\})}(y \mid \mathbf{x}) \right| \leq \epsilon L. 
\end{displaymath}
Therefore, for all $\mathbf{x} \in X^{m}$, for all $y \in Y$, 
\begin{displaymath}
e^{-\epsilon L} \leq \frac{p_{\A(T)}(y \mid \mathbf{x}) }{p_{\A(T \setminus \{t\})}(y \mid \mathbf{x})} \leq e^{\epsilon L} .
\end{displaymath}
That is, for all $t \in T$,
\begin{displaymath}
DTP_{\A,T}(t) \leq \epsilon \cdot L . 
\end{displaymath}
\end{proof}